\documentclass[12pt]{article}
\usepackage{fullpage}
\usepackage{balance}

\usepackage{amsmath,amsthm}
\usepackage{pgfplots}
\usepackage{comment}
\usepackage[utf8]{inputenc}
\usepackage{amsfonts}
\usepackage{mdframed}
\usepackage{framed}
\usepackage{hyperref}
\usepackage{mathtools}
\usepackage{authblk}
\usepackage{fancyhdr}
\usepackage{float}
\usepackage{mathdots}

\usepackage[noend]{algorithmic}
\usepackage[algoruled,vlined,nofillcomment,algo2e]{algorithm2e}
\usepackage{bm}

\usepackage{tikz}
\usetikzlibrary{calc,backgrounds}
\usetikzlibrary{arrows,calc,decorations.pathreplacing,positioning,shapes,shapes.geometric,shapes.callouts}
\usetikzlibrary{shapes,intersections,external}
\usepackage{pgfplots}

\newcommand{\dom}[1]{\mathbb{#1}}
\newcommand{\mech}[1]{\mathcal{#1}}
\newcommand{\tup}[1]{\vec{#1}}
\newcommand{\rv}[1]{\mathsf{#1}}

\newcommand{\trv}[1]{\tup{\rv{#1}}}

\renewcommand{\Pr}{\mathbb{P}}
\newcommand{\Ex}{\mathbb{E}}
\newcommand{\Var}{\mathbb{V}}
\newcommand{\EE}{\mathbb{E}}
\newcommand{\PP}{\mathbb{P}}
\newcommand{\TV}{\mathfrak{T}}
\renewcommand{\TV}{\mathrm{TV}}

\newcommand{\mse}{\mathrm{MSE}}
\newcommand{\SD}{\mathrm{TV}}
\newcommand{\bernoulli}{\mbox{\rm Ber}}
\newcommand{\uniform}{\texttt{Unif}}
\newcommand{\debias}{\texttt{DeBias}}

\newcommand{\polya}{\texttt{Polya}}

\newcommand{\Laplace}{\texttt{Lap}}
\newcommand{\discretelaplace}{\texttt{DLap}}
\newcommand{\fp}{{\mbox{\rm fp}}}

\newcommand{\poly}{\mathrm{poly}}

\newcommand{\N}{\mathbb{N}}
\newcommand{\R}{\mathbb{R}}
\newcommand{\Z}{\mathbb{Z}}

\newcommand{\remove}[1]{}

\usepackage{amsthm}

 \newtheorem{theorem}{Theorem}[section]
\newtheorem{lemma}{Lemma}[section]
\newtheorem{definition}{Definition}[section]
\newtheorem{corollary}{Corollary}[section]

\newtheorem{remark}{Remark}[section]
\usepackage[square,numbers]{natbib}

\title{Private Summation in the Multi-Message Shuffle Model}

\author[ ]{Borja Balle\footnote{Now at DeepMind.}}
\author[1]{James Bell}
\author[2]{Adri{\`a} Gasc{\'o}n\footnote{Work done while at the Alan Turing Institute.}}
\author[3]{Kobbi Nissim}

\affil[1]{The Alan Turing Institute}
\affil[2]{Google}
\affil[3]{Georgetown University}

\begin{document}

\maketitle

\begin{abstract}
The shuffle model of differential privacy (Erlingsson et al.\ SODA 2019; Cheu et al.\ EUROCRYPT 2019) and its close relative encode-shuffle-analyze (Bittau et al.\ SOSP 2017) provide a fertile middle ground between the well-known local and central models. Similarly to the local model, the shuffle model assumes an untrusted data collector who receives privatized messages from users, but in this case a secure shuffler is used to route messages from users to the collector in a way that hides which messages came from which user. An interesting feature of the shuffle model is that increasing the amount of messages sent by each user can lead to protocols with accuracies comparable to the ones achievable in the central model. In particular, for the problem of privately computing the sum of $n$ bounded real values held by $n$ different users, Cheu et al.\ showed that $O(\sqrt{n})$ messages per user suffice to achieve $O(1)$ error (the optimal rate in the central model), while Balle et al.\ (CRYPTO 2019) recently showed that a single message per user leads to $\Theta(n^{1/3})$ MSE (mean squared error), a rate strictly in-between what is achievable in the local and central models. 

This paper introduces two new protocols for summation in the shuffle model with improved accuracy and communication trade-offs. Our first contribution is a recursive construction based on the protocol from Balle et al.\ mentioned above, providing $\poly(\log \log n)$ error with $O(\log \log n)$ messages per user. The second contribution is a protocol with $O(1)$ error and $O(1)$ messages per user based on a novel analysis of the reduction from secure summation to shuffling introduced by Ishai et al.\ (FOCS 2006) (the original reduction required $O(\log n)$ messages per user).
We also provide a numerical evaluation showing that our protocols provide good trade-offs between privacy, accuracy and communication for realistic values of $n$.

\end{abstract} 
\maketitle

\section{Introduction}
\label{sec:intro}

In the shuffle model of differential privacy individuals communicate with an analyzer through a shuffler that (potentially) disassociates a message from its sender~\cite{BEMMRLRKTS17,erlingsson2019amplification, DBLP:journals/corr/abs-1808-01394}. Recent work on the shuffle model provides protocols for some basic statistical tasks (eg.\ averaging and computing histograms), demonstrating that the shuffle model allows for improved accuracy, compared with local differential privacy~\cite{erlingsson2019amplification,DBLP:journals/corr/abs-1808-01394, DBLP:conf/crypto/BalleBGN19,DBLP:journals/corr/abs-1906-08320,ghazi2019power,DBLP:journals/corr/abs-1909-11073,wang2019murs,erlingsson2020encode}.

A recurring problem in these papers is real summation, where the task is to approximate the sum of a collection of real values $x_1, \ldots, x_n \in [0,1]$ held by $n$ users.
The real summation problem has been extensively studied in the differentially privacy literature, eg.\ it is well-known that the optimal achievable errors in the central and local model are respectively $\Theta_{\epsilon}(1)$~\cite{DMNS06} and $\Theta_{\epsilon}(\sqrt{n})$~\cite{BNO08, ChanSS12}.
From a theoretical standpoint, this makes real summation a perfect candidate problem for exploring the differences between the shuffle model and the well-established central and local models, and for introducing new algorithmic techniques for the shuffle model. Furthermore, efficient and accurate protocols for real summation provide a key building block for innumerable applications ranging from simple statistical tasks (eg.\ computing means and variances of (bounded) numerical attributes across a population) to the training and evaluation of complex machine learning models (eg.\ computing empirical loss and averaging gradient updates in distributed settings).

The first protocol for real summation in the shuffle model, presented by Cheu et al.~\cite{DBLP:journals/corr/abs-1808-01394}, is a $O(\sqrt{n})$-message protocol with MSE (mean squared error) $O(1)$ where each user produces a unary encoding of their inputs with precision $O(\sqrt{n})$ and submits to the shuffler the result of applying independent binary randomized response mechanisms to each of the bits in this representation of the input. The analyzer simply sums the bits submitted by all users and then applies a debiasing operation. We note that while the protocol by Cheu et al.\ was presented in a model where all messages are sent through a single shuffler, their privacy analysis is based on the composition of several single-message protocols. Therefore an implementation of the protocol in a {\em parallel shuffle} model\footnote{See Section~\ref{sec:shuffleModelDef} for a formal definition of the single and parallel shuffle models.}, where multiple shufflers exists and each user sends a single message to each, would attain the same privacy guarantees.

Focusing on the setting where each user sends a single message, Balle et al.~\cite{DBLP:conf/crypto/BalleBGN19} provided matching upper and lower bounds showing real summation in the single-message shuffle model can be solved with MSE $\Theta(n^{1/3})$. In particular, this result shows that the single message shuffle model sits strictly in-between the local and central models of differential privacy. In terms of techniques, the protocol relies on amplification properties of the shuffle model introduced in~\cite{DBLP:conf/crypto/BalleBGN19} (see also \cite{erlingsson2019amplification}). The intuitive idea underlying this approach is to have clients send, along with their inputs, a set of random messages. The latter are mixed with the real inputs by the shuffler, providing a \emph{blanketing} noise that hides the contribution of any single user. This technique, originally introduced in \cite{DBLP:conf/crypto/BalleBGN19} for the single-message model where each user {\em either} contributes their real data or a message from the blanket, was later applied by Ghazi et al.~\cite{ghazi2019power} in the multi-message setting.

This work was in part previously made available in the form of two notes submitted to arXiv~\cite{balle2019differentially,balle2019improved}.

\subsection{Our Contributions}\label{sec:contributions}

This paper introduces new protocols for real summation in the shuffle model, improving on the accuracy and communication of previous protocols. A quick overview of our theoretical results and their relation with the prior works discussed above is provided in Table~\ref{tab:summation-ub}. In addition, Section~\ref{sec:numerics} presents a numerical comparison between our protocols and previous work in terms of privacy, communication and accuracy.
This work is partially based on results reported in the technical reports \cite{balle2019differentially,balle2019improved}.

Our first contribution is a recursive protocol based on the blanketing technique.
Using the single-message protocol from~\cite{DBLP:conf/crypto/BalleBGN19} as a building block,
we construct a multi-message single-round protocol where each subsequent message provides an estimation of the error incurred by approximating the sum using the previous messages. This protocol introduces a direct trade-off between the final accuracy and the number of messages $m$ sent by each user.
The crux of the analysis resides in balancing the accuracy boost and increased privacy loss incurred by each additional message.
The resulting protocol improves on the error achieved in the single message setting already for $m$ as small as $2$ and $3$ (where the MSE reduces to $O(n^{1/9})$ and $O(n^{1/27})$ respectively), and yields $\poly(\log \log n)$ error with $m=O(\log \log n)$.

Our second contribution is a generic reduction for obtaining private real summation protocols with optimal $O(1)$ MSE given a black-box protocol for secure summation over a finite group.
The reduction works by having each user discretize their inputs and then add a small amount of noise to them.
Secure summation over a finite group then allows us to simulate an optimal central mechanism (in this case, a discretized version of the celebrated Laplace mechanism \cite{DMNS06}) in a distributed manner.
The key challenge in realizing this reduction in the shuffle model resides in showing that a small number of shufflers can be used to implement secure summation over finite groups.

A starting point in this direction is the protocol of Ishai, Kushilevitz, Ostrovsky, and Sahai~\cite{ikos} (to which we refer henceforth as the {\em IKOS protocol}). In the setting of the IKOS protocol, $n$ users hold values $x_1,\ldots,x_n$ in a finite group $\mathbb{Z}_q$ and wish to securely compute their sum over $\mathbb{Z}_q$. To do so, each user $i$ splits their input $x_i$ into $m$ additive secret shares (over $\mathbb{Z}_q$), and sends all shares anonymously to the server via a shuffler. The server hence obtains the shuffled $nm$ shares and reconstructs the result by summing them up. Ishai et al.\ showed that splitting each value $x_i$ into $m = O(\log q + \sigma + \log n)$ shares
suffices to achieve statistical security parameter $\sigma$, in the sense that a (computationally unlimited) server seeing the set of shares submitted by the users cannot distinguish two inputs $(x_1,\ldots,x_n)$ and $(x'_1,\ldots,x'_n)$ with $\sum_i x_i = \sum_i x'_i$, except with advantage $2^{-\sigma}$.
Combined with our reduction, this gives a private real summation protocol with $O(1)$ MSE using $O(\log n)$ messages.

Puzzlingly, with the analysis provided in~\cite{ikos} the number of shares $m$ {\em grows} as a function of $n$. This is in contrast with the intuition that a larger number of participants should help any single user's contribution to ``hide in the crowd'', allowing $m$ to decrease as $n$ grows. 
The reduction from differentially private real-addition to secure addition over finite groups makes minimizing the number of shares $m$ a question of relevance not only to secure computation but also to understanding the trade-off between accuracy and communication in the shuffle model of differential privacy.

Our third contribution is a novel analysis of the IKOS protocol, showing that increasing the number of users does indeed help to reduce the number of messages required to attain a given level of security. We show that for statistical security parameter $\sigma$ it suffices to take the number of shares to be $m = \left\lceil\frac{2\sigma+\log_2(q)}{\log_2(n)-\log_2(e)}+2\right\rceil$. As a concrete example, computing the sum of $64$-bit numbers with security parameter $\sigma=80$ by $n=10^3$ users, our bounds show that $29$ messages per party suffice whereas $15$ messages per user suffice when $n=10^6$. In both cases, one of the messages can be sent outside the shuffle.

A more detailed account of our results is provided in Section~\ref{sec:results}, after the formal presentation of the shuffle model.

\paragraph{Comparing the two protocols.}
While our IKOS-based protocol gives optimal $O(1)$ MSE superior to the $\poly(\log\log n)$ MSE of the recursive protocol, there are some considerations that might make the latter preferable in some cases.
First, the recursive protocol can be instantiated with two and three messages -- yielding MSE $O(n^{1/9})$ and $O(n^{1/27})$ respectively -- whereas the IKOS protocol requires at least four messages even when the number of users asymptotically grows to infinity.
This shows there is a concrete advantage of going beyond one message, where the optimal MSE is $\Theta(n^{1/3})$~\cite{DBLP:conf/crypto/BalleBGN19}, and that this advantage increases with every additional message. In contrast, the analysis of the IKOS-based protocol does not provide any guarantees for such a small number of messages.

Furthermore, the IKOS-based protocol is not robust to manipulation by dishonest users. A single user deviating from the protocol can bias the result by an arbitrary amount\footnote{We leave open the question whether this can be remedied with a low penalty in efficiency.}. In comparison, our recursive protocol, does not put as much trust in the users, as the effect that any single dishonest user can have on the final sum is bounded by $1+O_\delta\left(\frac{(\log \log n)^2}{n^{2/3}\epsilon^2}\right)$. It is instructive to compare this result with the recent analysis of manipulation by Cheu et al.~\cite{CheuManipulation}, which shows that in the local model of differential privacy, a dishonest user's effect on the outcome is of magnitude $O(1/\epsilon)$.

Finally, it is interesting to observe that increasing communication has different effects on each of the protocols.
Intuitively, for a fixed choice of privacy parameter $\epsilon$, increasing the number of messages per user in the first protocol results in better accuracy, while in the second protocol it results in a smaller value for $\delta$.

\subsection{Concurrent Independent Work}\label{sec:concurrent}

We mention two concurrent and independent works which are relevant to this paper. First, independently of our work, Ghazi, Pagh, and Velingker~\cite{DBLP:journals/corr/abs-1906-08320} proposed protocols achieving MSE $O(1)$ with $O(\log n)$ messages, using an approach similar to that of the IKOS protocol. Their approach differs from ours in the employed distributed noise aggregation scheme: while Ghazi et al.\ rely on a similar technique to the one used by Shi et al.~\cite{shi}, we exploit the infinite divisibility properties of the geometric distribution, as suggested by Goryczka and Xiong~\cite{GX}. Quantitatively, our noise addition technique makes the error independent of $\delta$, saving a factor of $O(\log(1/\delta))$ in the MSE, and our communication complexity is independent of $\epsilon$ saving a factor of $O(\log(1/\epsilon))$ over~\cite{DBLP:journals/corr/abs-1906-08320}.

In a follow up work, Ghazi, Manurangsi, Pagh, and Velingker have, concurrently and independently of our work, obtained an analysis of the IKOS protocol that provides guarantees for a constant number of messages~\cite{DBLP:journals/corr/abs-1909-11073}.
Interestingly, their result gives the same asymptotic communication cost as our analysis, but the resulting constants are much larger\footnote{Their paper reads ``we
state our theorem in asymptotic notation and do not attempt to optimize the
constants in our proof''. The present constants are extracted from the proof of \cite[Theorem 3]{DBLP:journals/corr/abs-1909-11073}.}: for secure summation over $\Z_q$ with $n$ users, their analysis gives statistical security with parameter $\sigma$ when the number of messages is $m = \left\lceil \frac{100 (\sigma + \log_2(q))}{\log_2(n)-1} + 4\right\rceil$.
This difference in constants with the analysis developed in this paper can have a drastic effect in practical applications.
For example, in the case $\sigma = 80$, $q = 2^{64}$ and $n = 10^6$ their analysis requires $m = 765$ messages per user while ours requires only $m = 15$.

A second difference between our improved analysis of the IKOS protocol and the one provided by Ghazi et al.\ is that our analysis applies to summation over any finite group $\mathbb{Z}_q$ with no constraints on $q$, while the analysis by Ghazi et al.\ applies for summation over finite fields. The difference stems from the different proof techniques. While at its core our result involves an analysis of the distribution of connected components in a certain family of random graphs, their analysis is based on properties of the rank of a certain family of random matrices, and only works over finite fields because otherwise the relevant concepts from linear algebra (e.g.\ rank) are not defined.

\subsection{Other Related Work}\label{sec:relatedwork}

The shuffle model of differential privacy was originally motivated by the work of Bittau et al.\ \cite{BEMMRLRKTS17}, who proposed to introduce a trusted aggregator with very minimal functionalities to side-step the limitations of local differential privacy in the context of distributed data analysis.
The model was subsequently formalized in \cite{erlingsson2019amplification,DBLP:journals/corr/abs-1808-01394}.
With the goal of giving strong privacy guarantees for protocols that collect longitudinal data about the same set of users across a time horizon, Erlingsson et al.\ \cite{erlingsson2019amplification} analyzed privacy amplification in the context of an adaptive version of the shuffle model.
On the other hand, Cheu et al.\ \cite{DBLP:journals/corr/abs-1808-01394} initiated a systematic study of the relative power of the shuffle model compared with the local model by giving protocols for binary and real summation in the single and multi-message models, multi-message protocols for selection and histograms, and lower bounds for these tasks in the single-message setting. 

Subsequent works provide shuffle model protocols for other fundamental statistical tasks.
Beyond the problem of real summation which has been discussed at length above, these works focus on binary summation \cite{ghazi2020pure} as well as histograms, frequency estimation and private selection \cite{ghazi2019power,balcer2019separating}.
In addition, these papers have also studied the power of the shuffle model when compared to the local and central models in two separated axis: how increasing the number shuffled messages can improve the accuracy of the protocols \cite{ghazi2019power,balcer2019separating}, and whether asking for shuffle model protocols to provide pure differential privacy (ie.\ $\delta = 0$) has a significant effect on the accuracy achievable on a given task \cite{ghazi2020pure,balcer2019separating}.

On the systems side, \cite{wang2019murs,erlingsson2020encode} discuss how to construct practical shufflers satisfying the requirements of the theoretical shuffle model, and further discuss the intricacies of the threat models resulting from potential deployments of the shuffle model. Other implementation and experimental issues, including data fragmentation strategies when performing several tasks on the same dataset and the effects of flat or heavy-tailed datasets on empirical accuracy, are explored in \cite{erlingsson2020encode}.
 
\section{The Shuffle Model of Differential Privacy}
\label{sec:related}

This section formally defines the shuffle model used throughout the paper. This model provides privacy-preserving protocols for computing aggregate statistics using data held by users. The data is privatized and sent to a data collector for analysis through a trusted communication channel that randomly shuffles the messages of many users together.
We start by recalling the definition of differential privacy, proceed to introduce the shuffle functionality together with the rest of the ingredients necessary to specify a protocol in this model, define the semantics of the shuffle model used in this paper, and conclude by discussing related models.

\subsection{Central and Local Differential Privacy}

Differential privacy (DP) is a formal model providing strong individual privacy guarantees for data analysis tasks~\cite{DMNS06}. A randomized mechanism $\mech{M} : \dom{X}^n \to \dom{O}$ satisfies $(\epsilon,\delta)$-DP if for any pair of inputs $\tup{x}$ and $\tup{x}'$ differing in a single coordinate (we denote this relation by $\tup{x} \simeq \tup{x}'$) and any (measurable) event $E \subseteq \dom{O}$ on the output space we have
\begin{align*}
\Pr[\mech{M}(\tup{x}) \in E] - e^{\epsilon} \cdot \Pr[\mech{M}(\tup{x}') \in E] \leq \delta \enspace.
\end{align*}
In words: changing the data of one individual (among the $n$ contributing their data for analysis) does not significantly affect the probability distribution over outputs of the mechanism.

In this standard definition of differential privacy the mechanism has direct access to the data about $n$ individuals. Such data might be stored in a database curated by the party executing the mechanism, whose output is then released to the world. This is generally referred to as the \emph{central} (or \emph{curator}) model, and an implicit trust assumption is made that the curator will keep the database secret and only the outputs of DP mechanisms executed on the data will be released.

The \emph{local} model of differential privacy \cite{KLNRS08} considers a subclass of all DP mechanisms which operate under the assumption that each user applies randomness to privatize their own data before sending it to (a potentially untrusted) aggregator for analysis. In a \emph{non-interactive} setting where the aggregator receives the data from all users at the same time, mechanisms are specified in terms of a local randomizer\footnote{Each user could use a different randomizer, but here we ignore this setting for simplicity.} $\mech{R} : \dom{X} \to \dom{Y}$ and an analyzer $\mech{A} : \dom{Y}^n \to \dom{O}$, and output the result $\mech{M}_{\mech{R},\mech{A}}(\tup{x}) = \mech{A}(\mech{R}(x_1),\ldots,\mech{R}(x_n))$.
In this case the privacy guarantees of the mechanism are measured with respect to the view of the aggregator: $\mech{M}_{\mech{R},\mech{A}}$ is $(\epsilon,\delta)$-DP if for any $\tup{x} \simeq \tup{x}'$ and any event $E \subseteq \dom{Y}^n$ we have
\begin{align*}
\Pr[\tup{\mech{R}}(\tup{x}) \in E] - e^{\epsilon} \cdot \Pr[\tup{\mech{R}}(\tup{x}) \in E] \leq \delta \enspace,
\end{align*}
where $\tup{\mech{R}}(\tup{x}) = (\mech{R}(x_1),\ldots,\mech{R}(x_n))$.

\subsection{Shufflers and Multi-sets}
A shuffler $\mech{S} : \dom{Y}^n \to \dom{Y}^n$ is a randomized mechanism that returns the result of applying a uniform random permutation of $[n]$ to its inputs.
The shuffling operation ``erases'' all the information about the positions each particular message occupied in the input.
From an information-theoretic perspective this is equivalent to a deterministic mechanism $\mech{F} : \dom{Y}^n \to \dom{N}^{\dom{Y}}_n$ that outputs the multi-set of entries in an $n$-tuple (i.e.\ forgets about the order in which the inputs were provided).
This can be made formal by observing that the output of $\mech{S}(\tup{x})$ can be used to simulate the output of $\mech{F}(\tup{x})$ and vice versa.
Consequently, these two equivalent points of view on the action performed by the shuffler will be used interchangeably throughout the paper.

\subsection{The Shuffle Model}

\label{sec:shuffleModelDef}

The \emph{shuffle model} of differential privacy~\cite{erlingsson2019amplification, DBLP:journals/corr/abs-1808-01394} considers a data collector that receives messages from $n$ users (possibly multiple messages from each user). The key distinctive feature of the shuffle model is the assumption that a mechanism is in place to provide anonymity to each of the messages, i.e., in the data collector's view, the messages have been shuffled by a random unknown permutation. 
The concrete formulation of the shuffle model used in this paper is as follows (variants of this model are discussed below).

An $m$-message protocol in the shuffle model is specified by a pair of algorithms $\mech{P} = (\mech{R}, \mech{A})$, where $\mech{R}: \dom{X} \to \dom{Y}^m$, and $\mech{A}: (\dom{Y}^{n})^m \to \dom{O}$, for a number of users $n > 1$ and number of messages $m \geq 1$. We call $\mech{R}$ the \emph{local randomizer}, $\dom{X}$ the \emph{input space}, $\dom{Y}$ the \emph{message space} of the protocol, $\mech{A}$ the \emph{analyzer} of $\mech{P}$, and $\dom{O}$ the \emph{output space}.

The randomized mechanism $\mech{M}_{\mech{P}} : \dom{X}^n \to \dom{O}$ defined by the protocol $\mech{P}$ is defined as follows.
Each user $i$ holds a data record $x_i$, to which they apply the local randomizer to obtain a vector of messages $\vec{y}_i = (y^{(1)}_i, \ldots, y^{(m)}_i) = \mech{R}(x_i)$.
The users then use $m$ independent shufflers $\mech{S}^{(j)} : \dom{Y}^n \to \dom{Y}^n$, $j \in [m]$, to transmit their messages to the data collector; the $j$th message of each user is sent via the $j$th shuffler.
The data collector receives the outputs of the $m$ shufflers $\tup{y}^{(j)} = \mech{S}^{(j)}(y^{(j)}_1,\ldots,y^{(j)}_n)$, $j \in [m]$, and processes them using the analyzer to obtain the output of the protocol $\mech{M}_{\mech{P}}(\tup{x}) = \mech{A}(\tup{y}^{(1)}, \ldots, \tup{y}^{(m)})$ on inputs $\tup{x} = (x_1,\ldots,x_n)$. This is illustrated in Figure~\ref{fig:shuffle-model-diagram}.

From a privacy standpoint the data collector might have incentives to look at the outputs of the shufflers to obtain information about some user's input.
On the other hand, the shufflers are assumed to be fully trusted black-boxes free from any snooping or interference.
Thus, while the output of the protocol is the result of applying the analyzer to the shuffled messages, the privacy of such a protocol is formulated directly in terms of the \emph{view of the data collector} $\mech{V}_{\mech{P}}(\tup{x}) = (\tup{y}^{(1)}, \ldots, \tup{y}^{(m)})$ containing the tuples produced by the $m$ shufflers.
In particular, the protocol will provide privacy to individuals if for any pair of inputs $\tup{x}$ and $\tup{x}'$ differing in a single coordinate the views $\mech{V}_{\mech{P}}(\tup{x})$ and $\mech{V}_{\mech{P}}(\tup{x}')$ of the data collector are indistinguishable in the standard sense of differential privacy.

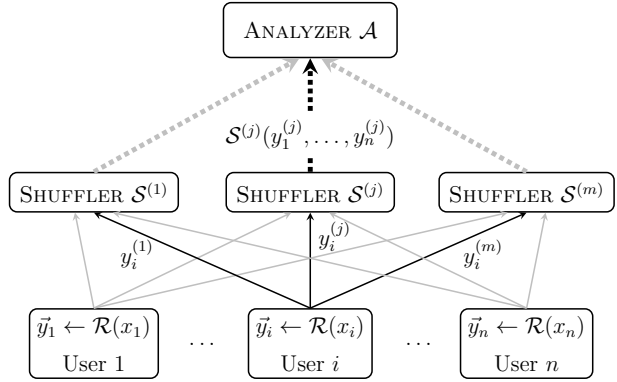
\begin{figure}
\begin{center}
\resizebox{0.5\textwidth}{!}{
\centering
  \begin{tikzpicture}
    \node[draw, fill=white, thick, rounded corners, inner sep=2ex] (analyzer) {
      \textsc{Analyzer} $\mech{A}$
    };

    \node[draw, fill=white, thick, rounded corners, inner sep=1ex] at ($(analyzer)+(0, -3)$) (shuffleri) {
      \textsc{Shuffler} $\mech{S}^{(j)}$
    };

    \node[draw, fill=white, thick, rounded corners, inner sep=1ex] at ($(shuffleri)+(-4, 0)$) (shuffler1) {
      \textsc{Shuffler} $\mech{S}^{(1)}$
    };

    \node[draw, fill=white, thick, rounded corners, inner sep=1ex] at ($(shuffleri)+(4, 0)$) (shufflerm) {
      \textsc{Shuffler} $\mech{S}^{(m)}$
    };
    
    \node[draw, fill=white, rectangle, rounded corners, thick, align=center] at ($(shuffleri)+(0, -2.8)$) (useri) {
      $\vec{y}_i \gets \mech{R}(x_i)$\\[1ex]
      User $i$
    };
    \node[draw, fill=white, rectangle, rounded corners, thick, align=center] at ($(useri)+(-4, 0)$) (user1) {
      $\vec{y}_1 \gets \mech{R}(x_1)$\\[1ex]
      User $1$
    };
    \node[draw, fill=white, rectangle, rounded corners, thick, align=center] at ($(useri)+(4, 0)$) (usern) {
      $\vec{y}_n \gets \mech{R}(x_n)$\\[1ex]
      User $n$
    };
    \node[draw=none] at ($(useri)+(2, 0)$) (dots1) {
      $\ldots$
    };
    \node[draw=none] at ($(useri)+(-2, 0)$) (dots2) {
      $\ldots$
    };

    \draw[draw=lightgray, ->, >=stealth, rounded corners, thick] ($(user1.north)+(0,0)$) -- ($(shuffler1.south)+(-.35, 0)$) node[pos=.2, above] {};
    \draw[draw=lightgray, ->, >=stealth, rounded corners, thick] ($(user1.north)+(0,0)$) -- ($(shuffleri.south)+(-.35, 0)$) node[pos=.2, above] {};
    \draw[draw=lightgray, ->, >=stealth, rounded corners, thick] ($(user1.north)+(0,0)$) -- ($(shufflerm.south)+(-.35, 0)$) node[pos=.2, above] {};

    \draw[->, >=stealth, rounded corners, thick] ($(useri.north)+(0,0)$) -- ($(shuffler1.south)+(0, 0)$) node[pos=0.8, below] {$y_i^{(1)}$};
    \draw[->, >=stealth, rounded corners, thick] ($(useri.north)+(0,0)$) -- ($(shuffleri.south)+(0, 0)$) node[pos=.75, right] {$y_i^{(j)}$};
    \draw[->, >=stealth, rounded corners, thick] ($(useri.north)+(0,0)$) -- ($(shufflerm.south)+(0, 0)$) node[pos=0.8, below] {$y_i^{(m)}$};

    \draw[draw=lightgray, ->, >=stealth, rounded corners, thick] ($(usern.north)+(0,0)$) -- ($(shuffler1.south)+(.35, 0)$) node[midway, sloped, above] {};
    \draw[draw=lightgray, ->, >=stealth, rounded corners, thick] ($(usern.north)+(0,0)$) -- ($(shuffleri.south)+(.35, 0)$) node[midway, sloped, above] {};
    \draw[draw=lightgray, ->, >=stealth, rounded corners, thick] ($(usern.north)+(0,0)$) -- ($(shufflerm.south)+(.35, 0)$) node[midway, sloped, above] {};
    
    \draw[draw=lightgray, ->, >=stealth, rounded corners, line width=.5ex, dotted] ($(shuffler1.north)+(0,0)$) -- ($(analyzer.south)+(-.2, 0)$) node[midway, right] {};
    \draw[->, >=stealth, rounded corners, line width=.5ex, dotted] ($(shuffleri.north)+(0,0)$) -- ($(analyzer.south)+(0, 0)$) node[pos=.33, fill=white] {$\mech{S}^{(j)}(y^{(j)}_1,\ldots,y^{(j)}_n)$};
    \draw[draw=lightgray, ->, >=stealth, rounded corners, line width=.5ex, dotted] ($(shufflerm.north)+(0,0)$) -- ($(analyzer.south)+(.2, 0)$) node[midway, right] {};
  \end{tikzpicture}
}
\caption{In the shuffle model, each local randomizer sends one message to each of $m$ independent shufflers.}\label{fig:shuffle-model-diagram}
\end{center}
\end{figure}
 
\begin{definition}
An $m$-message protocol $\mech{P} = (\mech{R}, \mech{A})$ in the shuffle model is $(\epsilon,\delta)$-DP if for any pair of inputs $\tup{x}, \tup{x}' \in \dom{X}^n$ with $\tup{x} \simeq \tup{x}'$ and any (measurable) event $E \subseteq (\dom{Y}^{n})^m$ we have
\begin{align*}
\Pr[\mech{V}_{\mech{P}}(\tup{x}) \in E] - e^{\epsilon} \cdot \Pr[\mech{V}_{\mech{P}}(\tup{x}') \in E] \leq \delta \enspace.
\end{align*}
\end{definition}

Note that the definition of view allows us to write the mechanism as the composition $\mech{M}_{\mech{P}} = \mech{A} \circ \mech{V}_{\mech{P}}$, and by the post-processing property of differential privacy the definition above implies that $\mech{M}_{\mech{P}}$ is an $(\epsilon,\delta)$-DP mechanism in the standard sense.
This definition also provides a clear parallelism between the shuffle model and the (non-interactive) local model of differential privacy: in the latter the view of the data collector is the tuple of messages produced by each user $(\tup{y}_1,\ldots,\tup{y}_n)$ and it is this view that is required to be differentially private.
In the shuffle model the view of the data collector is obtained from the view in the local model after the additional randomization step introduced by the shufflers. Quantifying the additional privacy provided by the shuffling step is one of the central questions in the shuffle model, which has given rise to a number of privacy amplification statements (both in implicit and explicit form) \cite{erlingsson2019amplification,DBLP:journals/corr/abs-1808-01394,DBLP:conf/crypto/BalleBGN19}. Broadly speaking, these results show that, in the one message ($m = 1$) case, if the protocol $(x_1,\ldots,x_n) \mapsto (\mech{R}(x_1), \ldots, \mech{R}(x_n))$ satisfies $\epsilon_0$-DP, then after shuffling the resulting protocol satisfies $(\epsilon,\delta)$-DP with $\epsilon = O(\min\{1,\epsilon_0\} e^{\epsilon_0} \sqrt{\log(1/\delta) / n})$. In particular, this implies that $\epsilon_0 = O(1)$ yield shuffle model protocols with $\epsilon = O_{\delta}(1/\sqrt{n})$ and $\epsilon_0 \leq \log(\sqrt{n})$ yields $\epsilon = O_{\delta}(1)$.

\subsection{Observations About the Model}\label{sec:observations-shuffle}
A number of variations on the concrete model described above could be considered.
Each of these raises one or more subtleties which we now discuss to further motivate the model used throughout the paper.

\paragraph{Single shuffler protocols.} Instead of $m$ parallel shufflers one can consider $m$-message protocols where the $m n$ messages contributed by all the users are shuffled together using a single shuffler.
From a privacy standpoint the single shuffler setting can provide additional privacy because, unlike in the case of parallel shufflers, the data collector cannot in general identify a group of $n$ messages containing at least one message from a given user $i$. This additional privacy gain, however, can only be realized when the $m$ messages generated by a user's local randomizer play ``exchangeable'' roles in the computation performed by the analyzer (e.g.\ when the messages are summed). This is the case, for example, for the multi-message protocol by Cheu et al.~\cite{DBLP:journals/corr/abs-1808-01394} for summation of inputs $x_i \in [0,1]$ using a single shuffler.
On the other hand, for the protocol we present in Section~\ref{sec:rec_protocol} it is crucial to apply a different debiasing operations to the messages coming from each of the shufflers.
We also note that our setting with $m$ parallel shufflers can be trivially simulated using a single shuffler where the $m$ messages contributed by each user come with a label $j \in [m]$ that enables the data collector to recover the message grouping after shuffling.
But even in cases where this simulation is a viable solution, there still exists a relevant distinction in terms of the assumed threat model which might make the parallel shuffler implementation more preferable: in the single shuffler setting an attacker only needs to compromise one shuffler to collect all the messages submitted by a given user, while in the parallel shuffler setting an attacker needs to compromise $m$ independent shufflers to obtain this same level of access.

\paragraph{User-dependent randomizers.} The model used in this paper assumes every user applies the same local randomizer to their data. This could be relaxed by letting each user employ a different randomizer. However, when used in its full generality this relaxation can lead to protocols where no privacy gains are obtained from the shuffling step. For example, in a protocol where the output domain of the randomizer of each user is disjoint from the output domains of the rest of randomizers. Note that each individual local randomizer can still provide local differential privacy in this case, but after shuffling it is possible to re-identify which messages were submitted by which user, so there is no privacy amplification coming from shuffling.

\paragraph{Interactive shuffle model.} The work of Erlingsson et al.\ \cite{erlingsson2019amplification} considers a more general shuffle model where each user's data can be processed by a different randomizer in a potentially adaptive fashion. This is achieved by significantly strengthening the trust assumptions on the shuffler, which now becomes an interactive intermediary between the analyzer and the users. In particular, in the interactive shuffle model users submit their data in plain text to the shuffler, who applies a random permutation to it, and then provides answers to a sequence of queries from the analyzer in the form of local randomizers which are applied to the next permuted data record. The sequence of randomizers provided by the analyzer can depend adaptively on previous answers. In this model, an attacker compromising the shuffler can, in principle, gain access to the plain-text data from the users -- while in the non-interactive model the users at least get some level of local differential privacy -- although some cryptographic constructions can be used to mitigate such risks \citep{BEMMRLRKTS17}. On the other hand, interactivity enables more complex computations which cannot be implemented in a single round non-interactive protocol (e.g.\ mini-batch stochastic gradient descent).

\section{Private Summation in the Shuffle Model}
\label{sec:results}

In this paper we are concerned with the problem of real summation where each user $i$ holds a real number $x_i \in [0,1]$ and the goal of the protocol is for the analyser to obtain a differentially private estimate of $\sum_{i=1}^n x_i$.
For each such protocol $\mech{P}$ we are interested in quantifying the final accuracy in terms of the worst-case \emph{mean squared error} (MSE) defined as
\begin{align}
\mse(\mech{P}) = \sup_{\tup{x} \in [0,1]^n} \Ex\left[\left(\mech{M}_\mech{P}(\tup{x}) - \sum_{i=1}^n x_i\right)^2\right] \enspace,
\end{align}
where the expectation is over the randomness in the protocol.
In some cases (e.g.\ Section~\ref{sec:rec_protocol}) we will use as building blocks summation protocols over different domains $\dom{X} \subset \R$ or $\dom{X} \subset \mathbb{Z}$.
For such protocols the definition of MSE is modified by replacing the supremum over $\tup{x} \in [0,1]^n$ to a supremum over $\tup{x} \in \dom{X}^n$.

Recall that in the curator model of differential privacy the Laplace mechanism provides an $\varepsilon$-DP mechanism for the problem of real summation with MSE $O_{\varepsilon}(1)$ and this is optimal.
On the other hand, the optimal MSE under $\varepsilon$-DP in the local model is $O_{\varepsilon}(n)$.
Real summation in the shuffle model raises interesting trade-offs between accuracy and communication, where the latter can be quantified both in terms of number of messages per user and size of these messages (in number of bits).
Known results for this problem together with the new contributions in this work are summarized in Table~\ref{tab:summation-ub}.
We now give a technical overview of our contributions.

\begin{table}
\begin{center}
\begin{tabular}{c|c|c}
Reference & MSE & Num.\ messages \\ \hline
\cite{DBLP:journals/corr/abs-1808-01394} & $O(\log^2(n/\delta)/\epsilon^2)$ & $O(\epsilon \sqrt{n})$ \\
\cite{DBLP:conf/crypto/BalleBGN19} & $O(n^{1/3}/\epsilon^{4/3})$ & 1 \\
\cite{DBLP:journals/corr/abs-1906-08320} & $O(\log(1/\delta) / \epsilon^2 )$ & $O(\log (n/ \epsilon \delta))$ \\
\cite{DBLP:journals/corr/abs-1909-11073} & $O(1/\epsilon^2)$ & $O(1 + \log_n(1/\delta))$ \\ \hline
Section~\ref{sec:rec_protocol} & $O((\log \log n)^2 \log(1/\delta) / \epsilon^2)$ & $O(\log \log n)$ \\
Section~\ref{sec:constant_error} & $O(1/\epsilon^2)$ & $O(\log (n/\delta))$ \\
Section~\ref{sec:improved_communication} & $O(1/\epsilon^2)$ & $O(1 + \log_n(1/\delta))$ \\
\end{tabular}
\caption{Comparison of protocols for real summation in the shuffle protocol.}
\label{tab:summation-ub}
\end{center}
\end{table}

\subsection{Our Results}

Our main contributions are two new protocols for real summation in the shuffle model with improved accuracy and communication over prior work. The two protocols are very different in nature. One is based on a recursive application of the optimal one message protocol from \cite{DBLP:conf/crypto/BalleBGN19} with a carefully designed finite-precision numeric representation scheme. The second protocol uses a novel analysis of a reduction from secure summation to multi-message shuffling based on random additive shares given in \cite{ikos} to simulate an optimal protocol in the central model using distributed noise.

\paragraph{Recursive protocol.} The recursive protocol (Section~\ref{sec:rec_protocol}) provides a trade-off between accuracy and number of messages: given $m$ messages and fixed privacy parameters $\epsilon$ and $\delta$, the protocol with $n$ users achieves MSE $O(n^{3^{-m}}(1 + m^3))$. In particular, this yields protocols with MSE $O(n^{1/9})$, $O(n^{1/27})$ and $O(\log \log^3 n)$ for $m = 2, 3, \lceil \log \log n \rceil$ messages respectively.

To give an overview of the recursive protocol we first recall that an optimal protocol for the one-message case (cf.\ \cite{DBLP:conf/crypto/BalleBGN19}) can be constructed by fixing a precision $p \in \N$, $p = O(n^{1/3})$, and letting each user $i$ apply to their input $x_i \in [0,1]$ a randomized response mechanism that with some probability $\gamma = O(n^{-2/3})$ returns a uniform value in $\{0,\ldots,p\}$ and with probability $1 - \gamma$ returns $\lfloor p x_i \rfloor + \bernoulli(p x_i - \lfloor p x_i \rfloor)$, the unbiased randomized rounding of $x_i$ to precision $p$. The values of $p$ and $\gamma$ are obtained by optimizing the trade-off between the error introduced by the rounding step and the error induced by the privacy-preserving randomized response step.

Our first protocol extends this approach to the setting where each user sends $m$ messages. The key idea is to capitalize on the privacy provided by shuffling the result of a randomized response mechanism by devising an encoding scheme for $x \in [0,1]$ into $m$ fixed-precision numbers and apply a randomized response mechanism to each of these numbers independently. The encoding proceeds as follows: we take a sequence of precisions $p_1, \ldots, p_m \in \N$ with $p_1 = O(n^{3^{-m}})$ and $p_{j+1} = \lceil p_j^{1/3} \rceil$, define the products $q_{j} = \prod_{l=1}^j p_l$, and approximate the input as
\begin{align*}
x \approx s_1 q_1^{-1} + s_2 q_2^{-1} + \cdots + s_m q_m^{-1}
=
\frac{s_1 + \frac{s_2 + \frac{\iddots + \frac{s_m}{p_m}}{p_3}}{p_2}}{p_1}
\end{align*}
with $s_j = \lfloor q_j x - p_j \lfloor q_{j-1} x \rfloor \rfloor$ and $q_0 = 0$. Applying a randomized rounding step to the last message $s_m$ yields an unbiased random approximation of the original input. Given such a representation of their input, a user then applies randomized response mechanisms to each of the $s_j$ to either submit the true message with probability $1 - \gamma_j$ or a uniform value in $\{0,\ldots,p_j\}$ with probability $\gamma_j = O(p_j / n)$. Upon receiving the messages from all users, the analyzer performs debiasing and summation operations to obtain the final result.

In our analysis of the protocol each message is allocated the same privacy budget and the final privacy analysis follows from the standard composition properties of differential privacy applied to the privacy analysis of the one-message protocol which relies on privacy amplification by shuffling. The parameters $p_j$ and $\gamma_j$ are chosen to optimize the error contributions due to privacy and rounding as discussed in the proof of Theorem~\ref{thm:rec_prot}. Refining this analysis using the advanced composition theorem leads to the result with $O((\log \log n)^2)$ MSE and $O(\log \log n)$ messages referenced in Table~\ref{tab:summation-ub} (Corollary~\ref{cor:rec_prot_adv_comp}).

\paragraph{IKOS protocol.} The IKOS protocol (Section~\ref{sec:improved_communication}) uses a constant number of messages to achieve the same asymptotic error as optimal mechanisms in the curator model. The main idea behind the protocol is to leverage the reduction from secure summation to multi-message shuffling of Ishai et al.\ \cite{ikos} in order to obtain a distributed implementation of a (discretized) Laplace mechanism using constant variance noise. Getting this idea to work required us to overcome a number of technical difficulties, including a significant strengthening of the reduction in \cite{ikos} which in its original form only gives a summation protocol with constant MSE by using a logarithmic number of messages as discussed in Section~\ref{sec:constant_error}.

The reduction in \cite{ikos} provides a protocol $\mech{P}$ where $n$ users, each holding an integer $x_i \in \Z_q$, $i \in [n]$, use $m$ independent shufflers to securely compute $\sum_i x_i$. This is achieved by letting each user generate $m$ additive shares of their input, i.e.\ a tuple $\tup{y}_i \in \Z_q^m$ of $m$ uniform random elements from $\Z_q$ conditioned on $\sum_j y_i^{(j)} = x_i$, and send the shares to a data collector using $m$ independent shufflers. By adding the shares received from all the shufflers, the data collector can exactly compute the sum of the inputs. The result of Ishai et al.\ is a security claim stating that if $\tup{x}, \tup{x}' \in \Z_q^n$ are input tuples with the same sum, $\sum_i x_i = \sum_i x'_i$, then the views $\mech{V}_{\mech{P}}(\tup{x})$ and $\mech{V}_{\mech{P}}(\tup{x}')$ of the data collector in the two executions are indistinguishable in the sense that $\TV(\mech{V}_{\mech{P}}(\tup{x}), \mech{V}_{\mech{P}}(\tup{x}')) \leq 2^{-\sigma}$ for some number of messages $m = O(\log(q n) + \sigma)$. One of our main technical contributions is an improved analysis of this reduction showing that the same security guarantees can be achieved with the much smaller number of messages $m = O(1 + \frac{\log(q) + \sigma}{\log(n)})$. Our analysis also provides small explicit constants for $m$ that make the result of practical interest.

Equipped with this improved reduction to perform secure summation in $\Z_q$ in the multi-message shuffle model with low communication, we construct a simulation of the optimal $\varepsilon$-DP real summation protocol $\mech{M}(\tup{x}) = \sum_{i} x_i + \Laplace(1/\varepsilon)$ by discretizing the input to a large enough abelian finite group $\Z_q$ and distributing the noise addition step across the $n$ users involved in the protocol. This is achieved by leveraging the observation from \cite{GX} showing that the discrete Laplace distribution over $\Z$ with distribution $\Pr[k] \propto \alpha^{-|k|}$ (i.e.\ the two-sided geometric distribution) is infinitely divisible, and therefore it is possible to obtain samples from this distribution by summing $n$ i.i.d.\ samples from some fixed distribution (in this case, the one obtained by taking the difference between two P{\'o}lya random variables). Since the secure summation protocol only works in a finite group, the users of our real summation protocol first discretize their inputs $x_i \in [0,1]$ to integers with precision $p = O(\sqrt{n})$, then add their share of the noise, and finally truncate the result modulo $q = O(n^{3/2})$. These parameter choices ensure that (i) the errors due to noise and discretization are of the same order of magnitude, and (ii) with high probability, there is enough space in the group to represent the noisy sum without overflows.

Putting the secure summation protocol together with the discretized distributed noise addition technique yields a protocol for private summation in the multi-message shuffle model. The privacy guarantees of this protocol follow from bounding its total variation distance from the discretized summation protocol with truncated discrete Laplace noise in the curator model (Lemma~\ref{lemma:SD2delta}). Thus, the protocol has  MSE $O(1)$ like the optimal protocol in the curator model, and uses $O(1)$ messages per user by virtue of our analysis of the summation to shuffling reduction.

\section{Recursive Protocol}
\label{sec:rec_protocol}

The intuition behind the private single-message summation protocol in~\cite{DBLP:conf/crypto/BalleBGN19} is simple: the local randomizer first represents the input $x$ in a discrete domain of size $p+1$, and then applies randomized response with probability $\gamma$ on the discretized (fixed-point) value. For example, if $x = 0.2342$ and $p=10$, the local randomizer submits $2 + \bernoulli(0.342)$ with probability $1-\gamma$, and a uniformly random value in $\{0,\ldots,p\}$ with probability $\gamma$. Here, the error of the fixed-point encoding due to the choice of $p$ needs to be balanced with the error due to the randomized response procedure. On one hand, $\gamma$ needs to grow with $p$, hence, for a  fixed choice of privacy parameters, the larger $p$ is the less often the true value will be reported. On the other hand, the error due to discretization decreases with $p$, as it corresponds to the precision in the fixed-point encoding.

In this section we show how to do better if we are allowed more messages. As before, assume that we choose the precision of our first message to be $p_1 = 10$. Then, the fixed-point encoding of $x=0.2342$ is $2$ and this has an error of $0.0342$, as before. Instead of applying a randomized rounding step as in the single-message case, we can just apply a randomized response of the value $2$, with domain $\{0,\ldots,10\}$ and probability $\gamma_1$ and {\em recursively} apply the same idea to the residual value $0.0342$. For example, if we choose $p_2 = 100$, and we are limited to just two messages, then the second message will be the randomized response of $34+ \bernoulli(0.2)$ with domain $\{0,\ldots,100\}$, and probability $\gamma_2$. In summary, we recursively apply the single-message procedure on the error of the fixed-point encoding of the current message, and randomized rounding in the last message, as the base case. This corresponds to choosing precisions $p_1, p_2, \ldots$ and probabilities $\gamma_1, \gamma_2, \ldots$ to, analogously to the single-message case, achieve a good balance between privacy and accuracy.

\begin{algorithm2e}[t]
  \DontPrintSemicolon
  \SetKwInput{KwPub}{Public Parameters}
  \SetKwComment{Comment}{{\scriptsize$\triangleright$\ }}{}
\caption{Local Randomizer $\mech{R}_{\gamma,p}$}\label{algo:lr}
        \KwPub{$\gamma,p$}
        \KwIn{$\bar{x} \in \{0,\ldots,p\}$}
        \KwOut{$y \in \{0,\ldots,p\}$}
\BlankLine
Sample $b\gets \bernoulli\left(\gamma\right)$\;
\eIf{$b = 0$}{
Let $y \gets \bar{x}$\;
}{
Sample $y \gets \uniform(\{0,\ldots,p\})$\;
}
\KwRet{$y$}
\end{algorithm2e}

\begin{algorithm2e}[t]
  \DontPrintSemicolon
    \SetKwInput{KwPub}{Public Parameters}
  \SetKwComment{Comment}{{\scriptsize$\triangleright$\ }}{}
\caption{Analyzer $\mech{A}_{\gamma,p,n}$}\label{algo:agg}
        \KwPub{$\gamma$ and number of parties $n$}
        \KwIn{Multiset $\{y_i\}_{i \in [n]}$, with $y _i\in \{0,\ldots,p\}$}
        \KwOut{$z \in \R$}
\BlankLine
Let $\hat{z} \gets \sum_{i=1}^n y_i$\;Let $z \gets \debias(\hat{z})$, where $\debias(w) = \left(w - \frac{n\gamma(p+1)}{2} \right) / \left(1-\gamma\right)$\;
\KwRet{$z$}
\end{algorithm2e}

The $m$-message recursive summation protocol is defined by $\mech{P}_{\vec{\gamma},\vec{p},n}^{\mathrm{rec}} = (\mech{R}_{\vec{\gamma},\vec{p}}^{\mathrm{rec}}, \mech{A}_{\vec{\gamma},\vec{p},n}^{\mathrm{rec}})$.
The local randomizer is shown in Algorithm~\ref{algo:rec_lr}. The algorithm takes a sequence of precisions $p_1,\ldots,p_m$ and probabilities $\gamma_1,\ldots,\gamma_m$. We also define $q_j=\prod_{l=1}^{j}p_l$ for simplicity. The algorithm
consists of $m$ executions of a discrete summation protocol (Algorithms~\ref{algo:lr} and~\ref{algo:agg}), where the $j$th message is the fixed-point encoding (with precision $p_j$) of the error up to message $j$ due to fixed-point encoding in previous messages, and the last message includes randomized rounding. 
This achieves the same goal than in the single-message protocol, i.e. obtain an unbiased estimate of the sum. The respective analyzers are given in Algorithms~\ref{algo:rec_agg} and~\ref{algo:agg} and only involve a standard debiasing step, and summation.

Let us first focus on the discrete summation subprotocol $\mech{P}_{\gamma,p,n} = (\mech{R}_{\gamma,p}, \mech{A}_{\gamma,p,n})$. Theorem~1 of \cite{DBLP:conf/crypto/BalleBGN19} states that this protocol is $(\epsilon,\delta)$-differentially private in the shuffle model, so long as $\epsilon\leq 1$ and $\gamma\geq \max\left\{\frac{14 p \log(2/\delta)}{(n-1)\epsilon^2},\frac{27 p}{(n-1)\epsilon}\right\}$.
The mean squared error on the result is $O_{\epsilon,\delta}(p^3)$, which is optimal for summation in $\{0,\ldots,p\}$ up to constants. This bound is given in (the proof of) Theorem~2 from \cite{DBLP:conf/crypto/BalleBGN19} as follows.

\begin{theorem}[\cite{DBLP:conf/crypto/BalleBGN19}]
  \label{thm:single_prot}
  Let $\epsilon\leq1$ and $\gamma = \max\{\frac{14 p \log(2/\delta)}{(n-1) \epsilon^2}, \frac{27 p}{(n-1) \epsilon} \} < 1$. Then
  \begin{equation*}
    \mse(\mech{P}_{\gamma,p,n})\leq \frac{n}{(1-\gamma)^2}\left(\frac{\gamma(p^2-1)}{12}+\frac{(p-1)^2\gamma(1-\gamma)}{4}\right) \enspace.
  \end{equation*}
  Neglecting $\log(1/\delta)$ terms, this is $O_\delta(p^3/\epsilon^2)$.
\end{theorem}

\begin{algorithm2e}[t]
  \DontPrintSemicolon
  \SetKwInput{KwPub}{Public Parameters}
  \SetKwComment{Comment}{{\scriptsize$\triangleright$\ }}{}
\caption{Local Randomizer $\mech{R}^{\mathrm{rec}}_{\vec{p},\vec{\gamma}}$}\label{algo:rec_lr}
        \KwPub{Number of shufflers $m$, vector of precisions $\vec{p}\in \mathbb{N}^m$, and vector of randomized response probabilities $\vec{\gamma}\in [0,1]^m$. Also, let $q_j=\prod_{l=1}^{j}p_l$ and $q_0 = 0$.}
        \KwIn{$x \in [0,1]$}
        \KwOut{$\vec{y} \in \left(\prod_{j=1}^{m-1}\{0,\ldots,p_j\}\right)\times \{0,\ldots,p_m+1\}$}
        \BlankLine
        \For{$j\gets 1$ \KwTo $m$}{
          Let $s_j \gets \left\lfloor q_jx- p_j\lfloor q_{j-1}x \rfloor \right\rfloor$\;
        }
        Let $r\gets \bernoulli\left(q_m x- p_m \lfloor q_{m-1} x \rfloor-s_m\right)$\;
        Let $s_m \gets s_m+r$\;
        \For{$j\gets 1$ \KwTo $m-1$}{
          Let $y_j \gets \mech{R}_{\gamma_j,p_j}(s_j)$
        }
        Let $y_m \gets \mech{R}_{\gamma_m,p_m+1}(s_m)$\;
        \KwRet{$\vec{y}$}
\end{algorithm2e}

\begin{algorithm2e}[t]
  \DontPrintSemicolon
    \SetKwInput{KwPub}{Public Parameters}
  \SetKwComment{Comment}{{\scriptsize$\triangleright$\ }}{}
\caption{Analyzer $\mech{A}^{\mathrm{rec}}_{\vec{\gamma},\vec{p},n}$}\label{algo:rec_agg}
        \KwPub{$\vec{\gamma}$,$\vec{p}$ and number of parties $n$}
        \KwIn{Multiset $\{\vec{y}_i\}_{i \in [n]}$, with $\vec{y} _i\in \left(\prod_{j=1}^{m-1}\{0,\ldots,p_j\}\right)\times \{0,\ldots,p_m+1\}$}
        \KwOut{$z \in [0,1]$}
\BlankLine
\For{$j\gets0$ \KwTo $m-1$}{
  Let $z_j \gets  \mech{A}_{\gamma_j,p_j,n}(((y_i)_j)_{i \in [n]})$\;
}
Let $z_m \gets \mech{A}_{\gamma_m,p_m+1,n}(((y_i)_m)_{i \in [n]})$\;
Let $z \gets \sum_{j=1}^m z_j/q_j$\;
\KwRet{$z$}
\end{algorithm2e}

To understand the privacy guarantees of the recursive algorithm first note that the privacy budget must be split over the different shufflers. It satisfies any differential privacy guarantee that is satisfied by the composition of $m$ copies of Algorithm \ref{algo:lr} run with parameters $\gamma_j$, $p_j$ and $n$. The condition $\epsilon\leq m$ could be relaxed by a constant factor at only a constant factor cost to the given bounds, it follows from the assumption $\epsilon\leq 1$ made in the single-message case to simplify the analysis. The following theorem states the main result of this section: $O(\log\log n)$ messages suffice to achieve MSE $O((\log\log n)^3)$. Proofs of this and other results in the paper are deferred to the appendix.

\begin{theorem}
  \label{thm:rec_prot}
  For $\epsilon\leq m$ and $\log(1/\delta)\geq 2\epsilon$, take $\epsilon_j=\epsilon/m$, $\delta_j=\delta/m$, $\gamma_j=\frac{14(p_j+\mathbb{I}_{j=m})\log(2/\delta_j)}{(n-1)\epsilon_j^2}$ and $p_j=n^{3^{j-m-1}}$. Then $\mech{P}^{\mathrm{rec}}_{\vec{\gamma},\vec{p},n}$ is $(\epsilon,\delta)$-differentally private and
  \begin{equation*}
    \mse(\mech{P}^{\mathrm{rec}}_{\vec{\gamma},\vec{p},n})=O_\delta\left(n^{3^{-m}}\left(1+\frac{m^3}{\epsilon^2}\right)\right) \enspace.
  \end{equation*}
  In particular, for $m=\lfloor \log_3(\log_2(n)) \rfloor$ and $p_j=2^{3^j}$ we get
  \begin{equation*}
    \mse(\mech{P}^{\mathrm{rec}}_{\vec{\gamma},\vec{p},n})=O_\delta\left(\frac{(\log \log n)^3}{\epsilon^2}\right) \enspace.
  \end{equation*}
\end{theorem}

The above theorem uses basic composition, using advanced composition gives the following corollary which follows from much the same proof.
The asymptotics in $n$ are better by a factor of $\log \log n$, however this comes at the cost of a $\log(1/\delta)$ factor and so in practice the basic composition analysis might be preferred depending on the concrete values of the parameters involved.

\begin{corollary}\label{cor:rec_prot_adv_comp}
  Let $m=\lfloor \log_3(\log_2(n)) \rfloor$, $p_j=2^{3^j}$, by taking appropriate $\epsilon_j=\Theta\left(\epsilon\sqrt{\frac{\log(1/\delta)}{m}}\right)$ and appropriate $\delta_j$ we get
    \begin{equation}
    \mse(\mech{P}^{\mathrm{rec}}_{\vec{\gamma},\vec{p},n})=O_\delta\left(\frac{(\log \log n)^2}{\epsilon^2}\right) \enspace.
  \end{equation}
\end{corollary}

In the next section we show a protocol that achieves constant error with a constant number of messages $m$, for large enough $n$. However, it requires $m > 3$, while the protocol presented in this section achieves MSE at most $O_{\epsilon,\delta}(n^{1/9})$ and $O_{\epsilon,\delta}(n^{1/27})$ with $2$ and $3$ messages, respectively. This is in contrast with the $\Omega(n^{1/3})$ lower bound for the single message case proved in \cite{DBLP:conf/crypto/BalleBGN19}.
 
\section{Constant Error from Secure Summation}
\label{sec:constant_error}
Extending the ideas from single message summation has failed to achieve an error that doesn't grow with $n$. In this section we take an alternative approach that will.
They key idea is to leverage a secure summation protocol in the shuffle model due to Ishai et al.~\cite{ikos}.
This protocol can be used to simulate (a discrete version of) the Laplace mechanism -- which provides $O(1)$ MSE in the central model -- in the shuffle model.
We will first define the secure summation functionality, and then provide our result in the form of a reduction from private summation to secure summation that preserves the number of messages being shuffled.
Combining this reduction with the result from \cite{ikos} we obtain a private summation protocol in the shuffle model with MSE $O(1)$ using $O(\log n)$ messages.
In the next section we provide an improved analysis of secure summation that yields a protocol using only $O(1)$ messages.

\subsection{Secure Exact Summation}\label{sec:secure_sum}

Let $\dom{G}$ be an Abelian group and $\mech{P} = (\mech{R},\mech{A})$ be a protocol in the shuffle model with input and output space $\dom{G}$. We say that $\mech{P}$ \emph{computes exact summation} if for any $\tup{x} \in \dom{G}^n$ we have $\mech{M}_{\mech{P}}(\tup{x}) = \sum_{i \in [n]} x_i$.

Such a protocol is deemed secure if the view of the aggregator on input $\tup{x}$ is indistinguishable from the view on another input $\tup{x}'$ producing the same output.
More formally, we say that a randomized protocol $\mech{P}$ computing exact summation over $\dom{G}$ provides \emph{worst-case statistical security} with parameter $\sigma$ if for any $\tup{x}, \tup{x}' \in \dom{G}^n$, such that $\sum_{i \in [n]} x_i = \sum_{i \in [n]} x_i',$ we have $\TV(\mech{V}_{\mech{P}}(\tup{x}),\mech{V}_{\mech{P}}(\tup{x}')) \leq 2^{-\sigma}$, where $\TV$ denotes the total variation (i.e.\ statistical) distance.

The following result from \cite{ikos} provides an upper bound on the number of messages required to compute exact summation securely in the shuffle model.

\begin{lemma}[\cite{ikos}]
\label{lemma:IKOS}
There exists an $n$-party protocol that computes summation over $\mathbb{Z}_q$ in the (single-shuffler) shuffle model with worst-case statistical security parameter $\sigma$ using $O(\log(qn)+\sigma)$ messages per party.  
\end{lemma}

The protocol by Ishai et al.\ is very simple.
For a given number of messages $m$ it uses the local randomizer $\mech{R}_{m}(x) = (\rv{Y}^{(1)}, \ldots, \rv{Y}^{(m)})$, where $\rv{Y}^{(j)}$ are uniform random variables in $\dom{G}$ conditioned on $\sum_{j \in [m]} \rv{Y}^{(j)} = x$. We refer to the output of $\mech{R}_{m}(x)$ as a collection of \emph{additive shares} of the input $x$. By using this randomizer, each user $i$ submits to the shuffler the additive shares $\mech{R}_{m}(x_i) = (\rv{Y}^{(1)}_i, \ldots, \rv{Y}^{(m)}_i)$ and all the analyzer needs to do is sum the messages after shuffling to recover the sum $\sum_{i \in [n]} x_i = \sum_{i \in [n]} \sum_{j \in [m]} \rv{Y}^{(j)}_i$.
The original presentation in \cite{ikos} assumes all $n m$ messages are shuffled together -- i.e.\ it is in the single-shuffler model discussed in Section~\ref{sec:observations-shuffle} -- but as the results in Section~\ref{sec:improved_communication} will show, the protocol also works in the $m$-parallel shufflers model considered in this paper.
We refer to this as the IKOS secure exact summation protocol, which is $\sigma$-secure as stated in the lemma.

We reproduce their proof in Section~\ref{sec:ikos_proof}, so that we can keep track of the constants. There we show that in fact it suffices to take
\begin{align}
\notag
m &= \frac{5}{2}\log(q)+\sigma+\log(n-1) \\
&\; + \frac{1}{4}\log(\log(q)+\sigma + \log(n-1))+O(1) \enspace.
\label{eqn:m_ikos}
\end{align}

\subsection{Distributed Noise Addition}
\label{sec:noise}
The challenge to simulate the Laplace mechanism assuming a secure summation protocol in the shuffle model is to distribute the noise addition operation across $n$ users.
A simple solution is to have a designated user add the noise required in the curator model. This is however not a satisfying solution as it does not withstand collusions and/or dropouts. To address this, Shi et al.~\cite{shi} proposed a solution where each party adds enough noise to provide $\epsilon$-DP in the curator model with probability $\log(1/\delta)/n$, which results in an $(\epsilon,\delta)$-DP protocol. However, one can do strictly better: the total noise can be reduced by a factor of $\log(1/\delta)$. To achieve this each party adds a discrete random variable such that the sum of the contributions
is exactly enough to provide $\epsilon$-DP.

A discrete random variable with this property is provided in~\cite{GX}, where it is shown that a discrete Laplace random variable can be expressed as the sum of $n$ differences of two P\'olya random variables (the P\'olya distribution is a generalization of the negative binomial distribution). Concretely, if $\rv{X}_i$ and $\rv{Y}_i$ are independent $\polya(1/n,\alpha)$ random variables then $\rv{Z}=\sum_{i=1}^n (\rv{X}_i-\rv{Y}_i)$ follows a \emph{discrete Laplace distribution} $\discretelaplace(\alpha)$ on $\mathbb{Z}$, i.e. $\PP[\rv{Z}=k]\propto \alpha^{|k|}$. This allows to distribute the noise in the Laplace mechanism across $n$ users and achieve $\epsilon$-DP on the result of the sum by tuning $\alpha$ appropriately.

This discrete noise distribution forms the basis of the randomizer of the protocol presented in the next section. When designing the protocol we will need to take care of working in the finite group $\mathbb{Z}_q$ instead of $\mathbb{Z}$, and thus analyze the potential effect on accuracy arising from overflows in the unlikely event that the noise becomes too large.

\subsection{Private Summation}
\label{sec:private_summation}

In this section we provide a reduction in the shuffle model for converting a secure integer summation protocol into a differentially private real summation protocol. We then combine this lemma with Lemma~\ref{lemma:IKOS} to derive an explicit protocol for differentially private real summation.
The privacy argument relies on leveraging the security properties of the IKOS summation protocol to compare the outputs of our private summation protocol in the shuffle model with a protocol in the central model where the noise is added to the sum of all the inputs and then an execution of the IKOS protocol is simulated.

\begin{lemma}
\label{lemma:privatefromsecure}
Suppose for any $q > 0$ there exists an $n$-party secure exact summation protocol over $\mathbb{Z}_q$ in the shuffle model providing worst-case statistical security with parameter $\sigma$ with $f(q,n,\sigma)$ messages per party. Then there exists an $(\epsilon, (1+e^\epsilon)2^{-\sigma})$-DP protocol for private real summation in the shuffle model with MSE $O_\epsilon(1)$ and $f(\lceil 2n^{3/2} \rceil,n,\sigma)$ messages per party.
\end{lemma}

\begin{algorithm2e}[t]
  \DontPrintSemicolon
    \SetKwInput{KwPub}{Public Parameters}
  \SetKwComment{Comment}{{\scriptsize$\triangleright$\ }}{}
\caption{Analyzer $\mech{A}_{n,m,p,q}$}\label{algo:agg_const}
        \KwPub{Number of parties $n$, number of messages per party $m$, precision $p$ and order $q>np$ of the additive group.}
        \KwIn{Multiset $\{y_i\}_{i\in [nm]}$, with $y _i\in \mathbb{Z}$}
        \KwOut{$z \in \mathbb{R}$}
\BlankLine
Let $z \gets \sum_{i=1}^{nm} y_i$ mod $q$ \Comment*[r]{Add all inputs mod $q$}
{\bf if }{$z>\frac{np+q}{2}$}{ \bf then }{$z \gets z-q$} \Comment*[r]{Correct for underflow}
\KwRet{$z/p$} \Comment*[r]{Rescale and return estimate}
\end{algorithm2e}

\begin{algorithm2e}[t]
\label{algo:locrand}
  \DontPrintSemicolon
  \SetKwInput{KwPub}{Public Parameters}
  \SetKwComment{Comment}{{\scriptsize$\triangleright$\ }}{}
\caption{Local Randomizer $\mech{R}_{\alpha,m,p,q}$}\label{algo:locrand_const}
        \KwPub{Noise magnitude $\alpha$, number of messages $m$, precision $p$ and order $q>np$ of the additive group.}
        \KwIn{$x \in [0,1]$}
        \KwOut{$\vec{y} \in [0..q-1]^m$}
\BlankLine
Let $\tilde{x}\gets \lfloor xp \rfloor + \bernoulli(xp-\lfloor xp \rfloor)$ \Comment*[r]{$\tilde{x}$ is the encoding of $x$ with precision $p$}
Let $y \gets \tilde{x}+\polya(1/n,\alpha)-\polya(1/n,\alpha)$ \Comment*[r]{add noise to $\tilde{x}$ as an element of $\Z_q$}
Sample $\vec{y} \gets \uniform([ 0..q-1 ]^m)$ conditioned on $\sum_{i\in [m]}y_i=y$

\KwRet{$\vec{y}$} \Comment*[r]{Submit $m$ additive shares of $y$}
\end{algorithm2e}

Combining Lemma~\ref{lemma:IKOS} and Lemma~\ref{lemma:privatefromsecure} we can conclude the following theorem.

\begin{theorem}
There exists an $(\epsilon,\delta)$-DP protocol in the shuffle model for real summation with MSE $O(1/\epsilon^2)$ and $O(\log(n/\delta))$ messages per party, each of length $O(\log n)$ bits.
\end{theorem}

Such a protocol can be constructed from the proofs of these lemmas and is given explicitly by taking the local randomiser $\mech{R}_{\alpha,m,p,q}$ given in Algorithm~\ref{algo:locrand_const}, and the analyzer $\mech{A}_{n,m,p,q}$ given in Algorithm~\ref{algo:agg_const}, with parameters $p=\sqrt{n}$, $q=\lceil 2np \rceil$, $\alpha=e^{-\epsilon/p}$ and $m$ as in Equation~\ref{eqn:m_ikos}.
 
\section{Communication Efficient Secure Summation}
\label{sec:improved_communication}
Next we provide a refined analysis of the statistical security provided by the IKOS secure summation protocol.
Our analysis works for the parallel shuffle model version of the protocol, which is a weaker assumption than the single-shuffler protocol considered in the original analysis (cf.\ Section~\ref{sec:observations-shuffle}).

Since the aggregator in the IKOS protocol is just summation of all messages, throughout this section we identify a protocol with the corresponding view of the aggregator for simplicity.
Accordingly, we identify the \emph{$m$-parallel IKOS protocol over $\dom{G}$} with the randomized map $\mech{V}_{m,n} : \dom{G}^n \to (\dom{G}^n)^{m}$ obtained as the view of the aggregator in an $m$-message protocol in the shuffle model with randomizer $\mech{R}_{m}$:
\begin{align*}
&\mech{V}_{m,n}(\tup{x}) = \Big(\mech{S}^{(1)}(\rv{Y}^{(1)}_1, \ldots, \rv{Y}^{(1)}_n), \ldots, \mech{S}^{(j)}(\rv{Y}^{(j)}_1, \ldots, \rv{Y}^{(j)}_n), \ldots \\
  & \qquad \ldots, \mech{S}^{(m)}(\rv{Y}^{(m)}_1, \ldots, \rv{Y}^{(m)}_n)\Big) \enspace,
\end{align*}
where $\mech{R}_m(x_i) = (\rv{Y}^{(1)}_i, \ldots, \rv{Y}^{(m)}_i)$.

The first result we prove uses the following weaker security notion compared to the worst-case definition from Section~\ref{sec:secure_sum}.
We say that a randomized protocol $\mech{V}$ with domain $\dom{G}^n$ provides \emph{average-case statistical security} with parameter $\sigma$ if we have $$\Ex_{\tup{\rv{X}},\tup{\rv{X}}'}[\TV_{|\tup{\rv{X}},\tup{\rv{X}}'}(\mech{V}(\tup{\rv{X}}),\mech{V}(\tup{\rv{X}}'))] \leq 2^{-\sigma} \enspace,$$ where $\tup{\rv{X}}$ and $\tup{\rv{X}}'$ are each an $n$-tuple of uniform random elements from $\dom{G}$ with the same sum, i.e $\sum_{i \in [n]} \rv{X}_i = \sum_{i \in [n]} \rv{X}_i'$. Note that $\tup{\rv{X}}$ and $\tup{\rv{X}}'$ are not assumed independent. Here $\TV_{|\tup{\rv{X}},\tup{\rv{X}}'}$ denotes the total variation distance when the inputs are fixed (i.e.\ over the randomness coming from the protocol).

\begin{theorem}[Average-case security]\label{thm:avg-security}
The protocol $\mech{V}_{m,n}$ over $\dom{Z}_q$ provides average-case statistical security with parameter
  \begin{equation*}
    \sigma = \frac{(m-1)(\log_2(n)-\log_2(e))-\log_2(q)}{2}
  \end{equation*}
  provided that $\sigma \geq 1$, $m\geq 3$ and $n\geq 19$.
\end{theorem}

While the above theorem only states average-case security, a simple randomization trick recovers worst-case security at the cost of one extra message per party (see Section~\ref{sec:random_inputs}).
Moreover, such a message does not need to be shuffled. This corresponds to a small variation on the parallel IKOS protocol where one of the messages contributed by each user is not sent through a shuffler; i.e.\ it is possible to unequivocally associate one of the messages from the input to each user.

We define the \emph{$m$-parallel IKOS protocol with randomized inputs over $\dom{G}$} as the randomized map $\tilde{\mech{V}}_{m,n} : \dom{G}^n \to (\dom{G}^n)^{m+1}$ obtained as follows.
Let $\mech{S}^{(j)} : \dom{G}^n \to \dom{G}^n$, $j \in [m]$, be $m$ independent shufflers.
For any $\tup{x} = (x_1,\ldots,x_n) \in \dom{G}^n$ define the random variables $(\rv{Y}^{(1)}_i, \ldots, \rv{Y}^{(m+1)}_i) = \mech{R}_{m+1}(x_i)$, $i \in [n]$, obtained by sampling $m+1$ additive shares for each input.
Then, the IKOS protocol with randomized inputs returns, for $j \in [m]$, the result of independently shuffling the $j$th shares of all the users together, concatenated with the $m+1$th \emph{unshuffled} shares:
\begin{align*}
  \nonumber
  &\tilde{\mech{V}}_{m,n}(\tup{x})  = \Big(\mech{S}^{(1)}(\rv{Y}^{(1)}_1, \ldots, \rv{Y}^{(1)}_n), \ldots, \mech{S}^{(j)}(\rv{Y}^{(j)}_1, \ldots, \rv{Y}^{(j)}_n), \ldots \\
  & \qquad \ldots, \mech{S}^{(m)}(\rv{Y}^{(m)}_1, \ldots, \rv{Y}^{(m)}_n), (\rv{Y}^{(m+1)}_1, \ldots, \rv{Y}^{(m+1)}_n)\Big) \enspace.
\end{align*}

\begin{corollary}\label{cor:wst-security}
The protocol $\tilde{\mech{V}}_{m,n}$ provides worst-case statistical security with parameter $\sigma$ given by the same expression as in Theorem~\ref{thm:avg-security} as long as $m \geq 3$ and $n \geq 19$.
  Thus, for fixed $n$, $q$ and $\sigma$, it suffices, for worst-case security, to take the number of \emph{shuffled} messages to be
  \begin{equation}
    m=\left\lceil\frac{2\sigma+\log_2(q)}{\log_2(n)-\log_2(e)}+1\right\rceil \enspace.
  \end{equation}
\end{corollary}

Therefore $m$ need not be taken to grow as $n$ goes to infinity as in Lemma~\ref{lemma:IKOS}. In fact, it turns out a larger $n$ results in each person having to send \emph{fewer} messages. This can be intuitively understood as a greater number of people providing a greater amount of ``cover'' to hide in the crowd.

Plugging this result into the private summation protocol from the previous section leads to our main result, which removes the $\log(n)$ factor from the number of messages submitted by each party.

\begin{theorem}
There exists an $(\epsilon,\delta)$-DP protocol in the shuffle model for real summation with MSE $O(1/\epsilon^2)$ and $O(\log(1/\delta))$ messages per party, each of length $O(\log n)$ bits.
\end{theorem}

\subsection{Proof Outline}
\label{sec:proof_outline}
We now provide an outline of the proof of Theorem~\ref{thm:avg-security}. The details can be found in the appendix. In this proof we will provide forward references to the required lemmas where we use them.

\begin{proof}[Proof of Theorem~\ref{thm:avg-security}]
  Lemma~\ref{lem:single_shuffling} in Section~\ref{sec:single_shuffling} says that protocol $\mech{V} = \mech{V}_{m,n}$ provides statistical security with distance bounded by the following expression:
  \begin{equation*}
    \Ex_{\tup{\rv{X}},\tup{\rv{X}}'}[\TV_{|\tup{\rv{X}},\tup{\rv{X}}'}(\mech{V}(\tup{\rv{X}}),\mech{V}(\tup{\rv{X}}'))] \leq \sqrt{q^{m n-1} \Pr[E] - 1} \enspace,
  \end{equation*}
where $E$ is an event specified in Section~\ref{sec:single_shuffling} that relates two independent executions of $\mech{V}$ with the same (random) inputs.
  To bound the probability of event $E$, Section~\ref{sec:graph_argument} defines a probability distribution over a certain class of multigraphs. Lemma~\ref{lem:graph_argument} then says that if $G$ is drawn from this distribution and $C(G)$ is the number of connected components in $G$, then
  \begin{equation*}
    \Pr[E] \leq q^{-m n}\Ex[q^{C(G)}]\enspace.
  \end{equation*}
  This expectation is then bounded in Lemma~\ref{lem:bound_on_graph}, which says that, if $n\geq 19$, $m \geq 3$ and $q\leq \frac{1}{2}\left(\frac{n}{e}\right)^{m-1}$, 
  \begin{equation*}
    \Ex[q^{C(G)}]\leq q+q^2\left(\frac{n}{e}\right)^{1-m}\enspace.
  \end{equation*}
  Note that the condition required on $q$ here is implied by
  \begin{equation*}
    \frac{(m-1)(\log_2(n)-\log_2(e))-\log_2(q)}{2}\geq 1
  \end{equation*}
  and thus follows from the condition in the theorem that $\sigma\geq 1$.

  Putting this together we get average case statistical security less than or equal to $\sqrt{q(e/n)^{m-1}}$. Thus we have average case statistical security $2^{-\sigma}$ for
  \begin{equation*}
    \sigma=\frac{(m-1)(\log_2(n)-\log_2(e))-\log_2(q)}{2}\enspace.
  \end{equation*}
\end{proof}

\subsection{Reduction to Random Inputs}
\label{sec:random_inputs}
To obtain Corollary~\ref{cor:wst-security} from Theorem~\ref{thm:avg-security} we show that a certain level of average-case security with $m$ messages implies the same level of worst-case security with $m+1$ messages.
In addition, we show that the additional message required to reduce worst-case security to average-case security does not need to be sent through a shuffler.
Note that the expression for the required $m$ in Corollary~\ref{cor:wst-security} is a simple rearrangement of the expression for $\sigma$, so the following lemma gives the desired result.

\begin{lemma}\label{lem:avg_to_worst}
If $\mech{V}_{m,n}$ provides average-case statistical security with parameter $\sigma$, then $\mech{V}_{m+1,n}$ and $\tilde{\mech{V}}_{m,n}$ provide worst-case statistical security with parameter $\sigma$.
\end{lemma}

\section{Numerical Experiments}
\label{sec:numerics}
\begin{figure*}[t]
\hfill
  \includegraphics[width=0.24\textwidth]{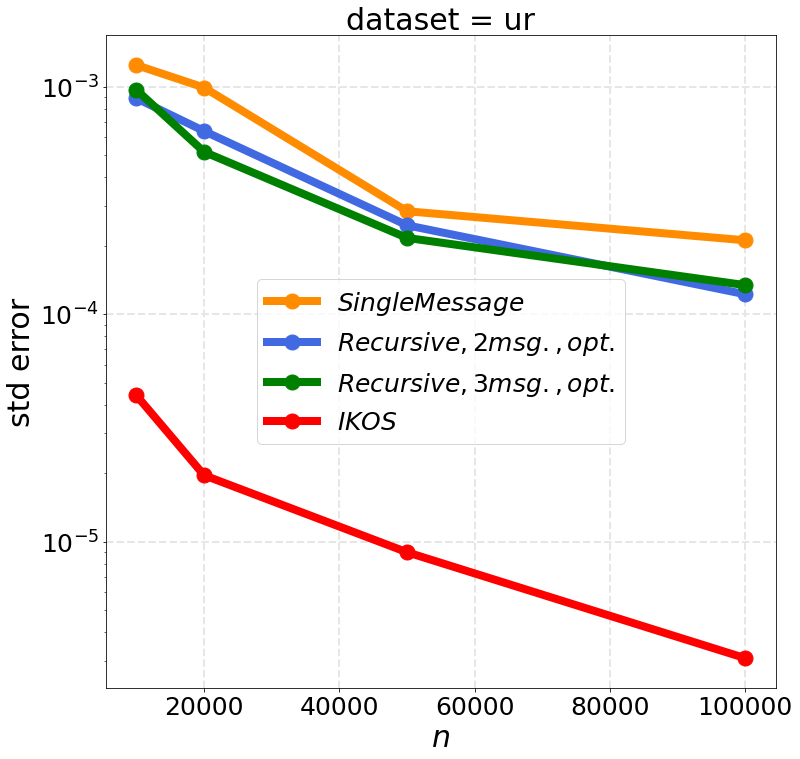}
\hfill
  \includegraphics[width=0.24\textwidth]{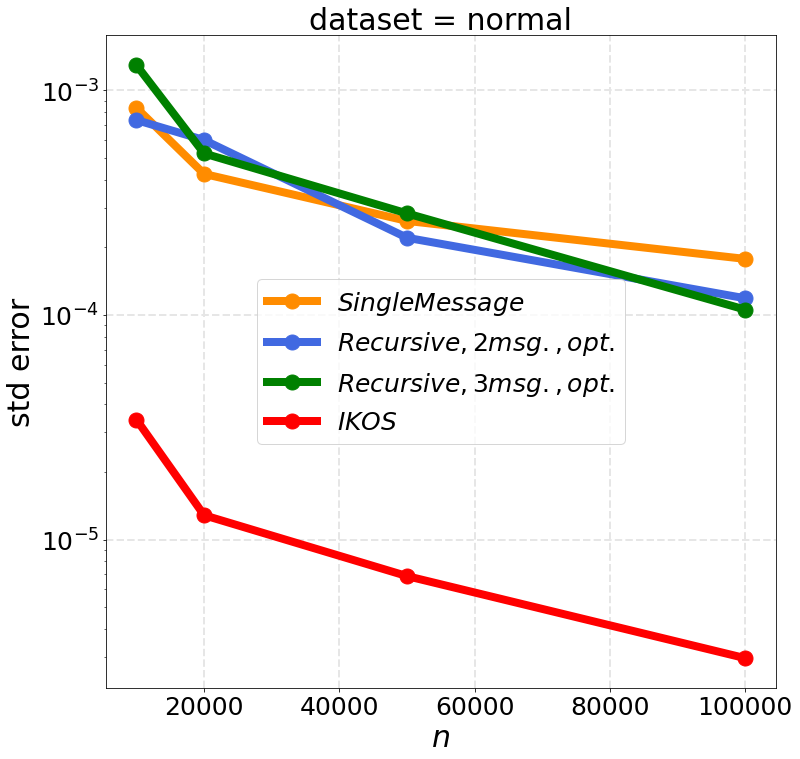}
\hfill
  \includegraphics[width=0.24\textwidth]{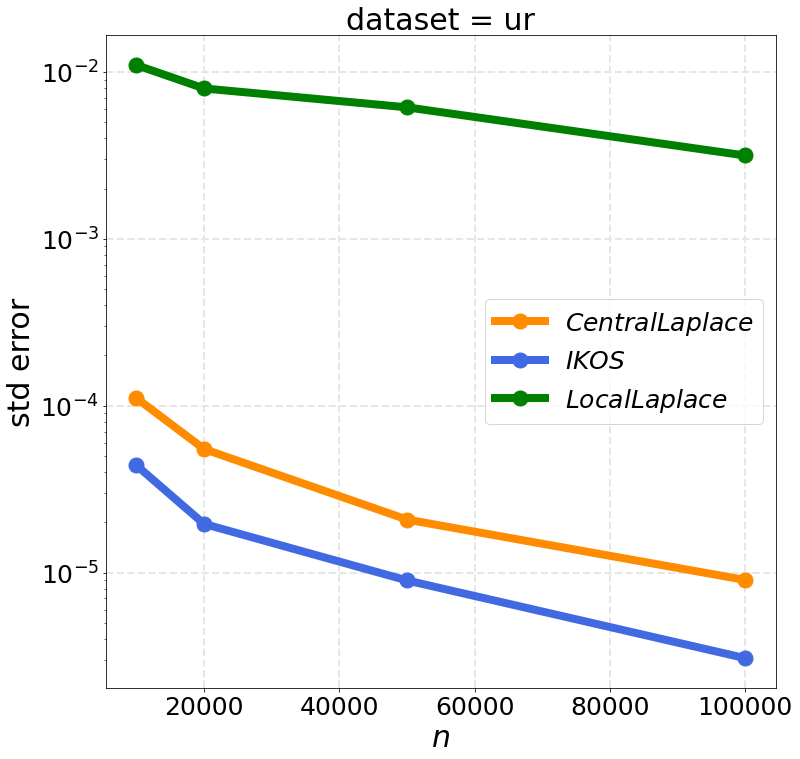}
\hfill
  \includegraphics[width=0.24\textwidth]{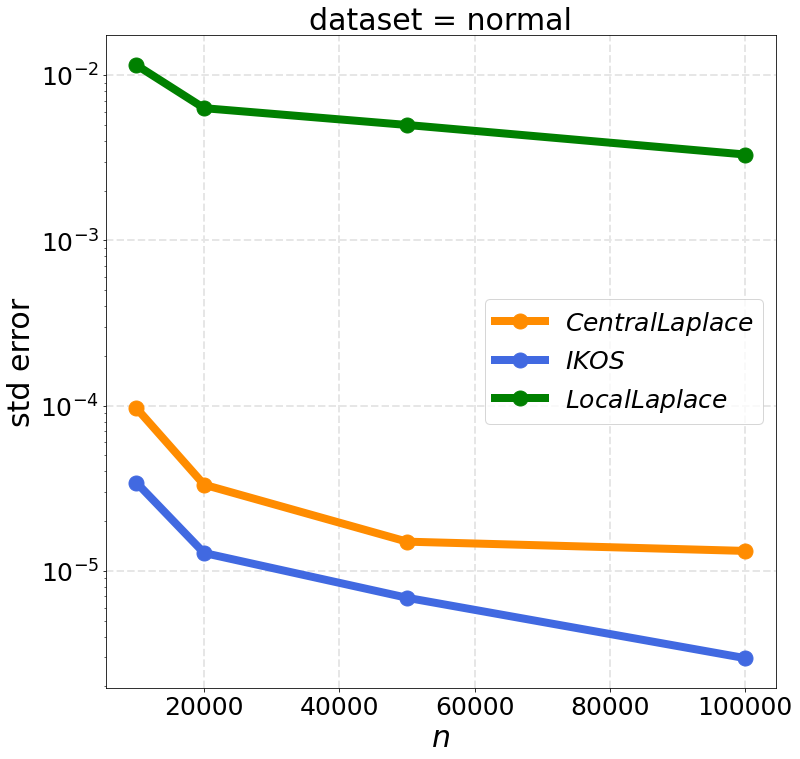}
\hfill
  \caption{Error in mean estimation for $n \in \{10^3, 2\cdot 10^3, 5\cdot 10^3, 10^4\}$ with $\epsilon = 1$ and $\delta = 1/n^2$ on synthetic ``ur'' and ``normal'' datasets. Two leftmost plots compare the single-message shuffle protocol from~\cite{DBLP:conf/crypto/BalleBGN19} with the multi-message protocols proposed in this paper. The two rightmost plots compares our best protocol in the shuffle model (IKOS) with standard procotols in the central and local model (LaplaceCentral, LaplaceLocal). Results are averaged over 20 runs.}
    \label{fig:plots}
\end{figure*}

We performed two sets of numerical evaluations of the protocols presented in the paper.
The first set of evaluations compares the \emph{communication and MSE bounds} of our protocols and other protocols from the literature as a function of the number of users and the privacy parameters. Due to space reasons these results are reported in Appendix~\ref{sec:more-experiments}.
Here we report the results of the second set of experiments, where we evaluate the \emph{empirical accuracy} of several protocols for the task of real summation on real and synthetic datasets.

In particular. we implemented the mechanism corresponding to each of our protocols, as well as their numerically optimized variants (see Appendix~\ref{sec:numerics_details} for details),
and evaluated them for the task of computing the average of $n$ normalized $x_i \in [0,1]$.
This corresponds to the task of frequency estimation, which covers a wide range of application as discussed in the
introduction (including private evaluation of ML models, as well as estimating mean and variance of a population).
In these experiments we report standard error in estimating that quantity, defined as $\frac{1}{n}|\sum_i x_i - f(x_1, \ldots, x_n)|$,
where $f$ is the private statistic outputted by each of our protocols.
Let us remark that we are estimating
{\em mean} and thus the standard error that we should expect for our protocols corresponds to the
square root of the values reported in Table~\ref{tab:numerics} in the appendix, divided by the dataset size $n$.

\paragraph{Baseline Protocols.}
Besides our protocols, we also consider three baselines: The single-message protocol from Balle et al.~\cite{DBLP:conf/crypto/BalleBGN19} (as a reference of the accuracy one can achieve in the shuffle model
with a single message per user), the Laplace mechanism (which we denote ``LaplaceCentral'') as it would be calibrated in the central model, i.e. $\sum_i x_i + Laplace(1/\epsilon)$
(as a reference of the accuracy one can achieve in the central model), and
the Laplace mechanism as it would be used in the local model $\sum_i (x_i + Laplace(1/\epsilon))$ (which we denote ``LaplaceLocal'') (as a reference of the accuracy one can achieve in the local model).
We evaluate our protocols in both synthetic and real-world data.

\paragraph{Synthetic Experiments.}
We considered both a dataset of uniformly random samples from $[0,1]$ (which we denote ``ur'') and a sample from a normal distribution with mean $0.573$
and standard deviation $.1$, which we refer to as ``normal''. Accuracy results as a function of the number of users are displayed in Figure~\ref{fig:plots} for both datasets.
A set of plots provides a comparison between the different protocols in the shuffle model, while the other set of plots compares the best protocol in the shuffle model with standard protocols in the central and local model.
The same trends are observed for both types of synthetic data.
As expected, the obtained accuracies are clustered according to their analytical errors,
which in turn correspond to the model they operate in. It is worth mentioning that IKOS
gives better error than CentralLaplace in these datasets. We speculate that this is due to the fact that the Geometric mechanism
works well with counts, as it is discrete.

\paragraph{Real-world data Experiments.} As a real-word dataset we used the Adult dataset~\cite{kohavi1996scaling}. This dataset contains $n = 32561$ curated records from the 1994 US Census database.
This is a dataset commonly used evaluate classification algorithms. We focused on the task of differentially private
estimation the mean of the ``age'' attribute normalized by the maximum age (with parameters $\epsilon = 1, \delta = 1/n^2$).
Table~\ref{tab:adult} reports the mean and standard deviation of our estimate over $20$ runs of several protocols.
As expected, IKOS incurs error close to the one of CentralLaplace, and LocalLaplace has significantly worse error than the rest.
Our recursive protocols outperform the single-message protocol from~\cite{DBLP:conf/crypto/BalleBGN19} in this task, and
numerically optimized variants outperform their analytical counterparts.

Overall, experiments show that our protocols give a good trade-off between accuracy and privacy, for realistic values $n$. Furthermore, the results in Appendix~\ref{sec:more-experiments} show the advantages of our analyses in improving the communication vs.\ accuracy trade-off of protocols in the shuffle model with respect to previous works.

\begin{table}
 \begin{center}
 \begin{tabular}{ccc}
 \hline
 Algorithm               &     Mean &   Std. dev. \\
\hline
 CentralLaplace          & $3.53\cdot 10^{-5}$ &    $3.21\cdot 10^{-5}$ \\
 LocalLaplace            & $4.28\cdot 10^{-3}$ &    $3.47\cdot 10^{-3}$ \\
 IKOS                    & $7.5\cdot 10^{-6}$  &    $6.17\cdot 10^{-6}$ \\
 SingleMessage           & $6.65\cdot 10^{-4}$ &    $4.01\cdot 10^{-4}$ \\
 Recursive, 2 msg.       & $4.58\cdot 10^{-4}$ &    $3.67\cdot 10^{-4}$ \\
 Recursive, 2 msg., opt. & $4.12\cdot 10^{-4}$ &    $2.34\cdot 10^{-4}$ \\
 Recursive, 3 msg.       & $5.25\cdot 10^{-4}$ &    $3.79\cdot 10^{-4}$ \\
 Recursive, 3 msg., opt. & $3.7\cdot 10^{-4}$  &    $2.66\cdot 10^{-4}$ \\
\hline
\end{tabular}
\caption{Mean and standard deviation (over $20$ runs) of the standard error of several protocols on Adult ($n = 32561, \epsilon = 1, \delta = 1/n^2$).}
\label{tab:adult}
\end{center}
\end{table}

\section{Discussion}
\label{sec:conclusion}
We have presented two protocols resulting from two different approaches to real summation in the shuffle model of differential privacy. The first protocol achieves MSE as small as $O((\log\log n)^2)$, and to do so requires each user to send $O(\log\log n)$ messages of size $O(\log n)$ each. The second protocol achieves constant error, and to do so only requires a constant number of messages per user. The first protocol uses the blanketing idea from Balle et al.~\cite{DBLP:conf/crypto/BalleBGN19}, while the second protocol relies on a reduction from secure discrete summation to shuffling by Ishai et al.~\cite{ikos}. In fact, the core of our contribution in the design of the second protocol is in the improved analysis of such reduction, which is of independent interest from a secure computation perspective. The fact that the ideas behind the two protocols are so different has interesting implications worth discussing, as the two results complement each other in interesting ways.

{\em Breaking the $O(n^{1/3})$ barrier.} Previous works by Balle et al.~\cite{DBLP:conf/crypto/BalleBGN19} and Cheu et al.~\cite{DBLP:journals/corr/abs-1808-01394} show that with a single message one could achieve MSE $\Theta(n^{1/3})$, and $O(n^{1/2})$ messages where enough to achieve error $O(1)$. Our work reduces this gap in several ways: first, we showed that the IKOS protocol with the proof technique proposed by Ishai et al.\ in the original paper leads to a protocol where $O(\log n)$ messages are enough to achieve constant error. Next, we further improve on that by providing an alternative proof which resolves affirmatively the question of whether a constant number of messages is enough to achieve constant error for summation in the shuffle model. Moreover, our proof tracks constants, which allows us to provide concrete bounds on the number of messages required to achieve constant error, for given values of $n$, $\epsilon$ and $\delta$. However, this result assumes that the number of messages per user is at least $3$, and thus does not address the question of whether $2$ messages are enough to break the lower bound by Balle et al. The recursive protocol, on the other hand, allows to trade accuracy by number of messages, and resolves this question affirmatively: the two-message variant of this protocol has MSE $O(n^{1/9})$. Whether that upper bound is tight is an open question.

{\em Robustness against adversarial behaviour.} An important difference between our two protocols is their ability to withstand adversarial users that might deviate from the prescribed protocol execution to bias (or even completely spoil) the result. To protect from this kind of manipulation attacks, protocols must have mechanisms to limit the influence of any particular user (or a possibly coordinated coalition of them) beyond what the function being computed inherently allows. In particular, for the task of real summation, a particular user's influence is at least $1$. Note that the maximum possible influence of a particular user in our IKOS protocol from Section~\ref{sec:improved_communication} is $2n$, as a dishonest user could replace their actual input (which should be encoded as an integer in $\{0,\ldots,p\}$ as prescribed by the protocol) by an arbitrary value in $\dom{Z}_q$. We leave open the question of whether this protocol can be efficiently ``patched'' to address this issue. Recent work by Cheu et al.~\cite{CheuManipulation} studies the issue of manipulation by coalitions of dishonest users in the context of the local model of differential privacy. Their results show that non-interactive protocols in the local model are in a sense manipulable. In particular, for the case of summation a dishonest user can skew the result by $\Omega(1/\epsilon)$. Intuitively, this is because in the local model (where each user gets privacy independently of the rest) large changes in the input distribution induce only a small change in the distribution of messages observed by the analyzer. Hence, dishonest users inducing small changes in the distribution of messages can induce large changes in the final estimate produced by the analyzer. In contrast, for our recursive protocol from Section~\ref{sec:rec_protocol} the maximum skew a user can induce is $1+O_\delta\left(\frac{(\log \log n)^2}{n^{2/3}\epsilon^2}\right)$. To see this, note that messages sent by a dishonest user, even if they are anonymized via the shuffer, still need to ``typecheck''. Hence, even a dishonest user must report $m$ values $y_1,\ldots, y_{m}$. These values must be in $\{0,\ldots,p_j\}$ and $\{0,\ldots,p_m+1\}$, with precisions $p_1,\ldots, p_{m}$ prescribed by the protocol. Similarly to what happens in the local model, our analyzer debiases the result by scaling each contribution $y_j$ by $1/(1-\gamma_j)$. As $\gamma_j$ in our protocol is $O_\delta\left(\frac{1}{n^{1-3^{j-m-1}}\epsilon_j^2}\right)\leq O_\delta\left(\frac{m^2}{n^{2/3}\epsilon^2}\right)$, the amount by which each $y_j$ is scaled is $1+O_\delta\left(\frac{(\log \log n)^2}{n^{2/3}\epsilon^2}\right)$. As even without privacy a single user can skew the result by $1$ this represents very little extra manipulability compared to the non-private case.

It is important to remark that dishonest users in the shuffle model can harm privacy by not providing blanketing noise as expected. This is easily achieved by simply sending invalid messages. A straightforward defense is to increase the overall noise by a constant factor that accounts for the maximum fraction of expected misbehaving users.
  
\section*{Acknowledgements}
A.~G.\ was supported by the EPSRC grant EP/N510129/1, and funding from the UK Government’s Defence \& Security Programme in support of the Alan Turing Institute.
J.~B.\ was supported by the EPSRC grant EP/N510129/1, and funding from the UK Government’s Defence \& Security Programme in support of the Alan Turing Institute.
K.~N.\ was supported by NSF grant no.~1565387, TWC: Large: Collaborative: Computing Over Distributed Sensitive Data. Work partly done while K.~N.\ was visiting the Alan Turing Institute.
We thank Samson Zhou for pointing out a mistake in the proof of Lemma~\ref{lem:avg_to_worst} in a previous version.

\bibliographystyle{plainnat}
\bibliography{main}

\newpage

\appendix

\section{Analysis of Recursive Protocol}
\label{sec:rec_proof}
We start by proving the following lemma on the error that the protocol incurs, which gives a bounds on the MSE of the real summation protocol in terms of bounds on the MSE of the discrete summation subroutines we viewed as standalone summation protocols for inputs on $\{0,\ldots,p_j\}^n$.

\begin{lemma}\label{lemma:bound-outer-mse}
    Let $B_j \leq \mse(\mech{P}_{\gamma_j,p_j+\mathbb{I}_{j=m},n})$, for j in $[m]$. Then,
  \begin{equation*}
    \mse(\mech{P}^{\mathrm{rec}}_{\vec{\gamma},\vec{p},n})\leq\frac{n}{4q_m^2}+\sum_{j=1}^m\frac{B_j}{q_j^2} \enspace.
  \end{equation*}
\end{lemma}

\begin{proof}
  Note in the local randomizer that $\sum_j s_j/q_j$ is an unbiased estimate of $x_i$ with variance equal to the variance of $r$ divided by $q_m^2$ which is at most $1/(4q_m^2)$. Summing this noise over all local randomizers gives mean squared error $n/(4q_m^2)$. The contribution of the call to the non recursive local randomizer on $s_j$ is scaled by the corresponding $1/q_j$ providing the sum in the lemma.
\end{proof}

\begin{proof}[Proof of Theorem \ref{thm:rec_prot}]
  For $\gamma_j=\frac{14(p_j+\mathbb{I}_{j=m})\log(2/\delta_j)}{(n-1)\varepsilon_j^2}$, so long as $14\log(2/\delta_j) \geq 27\varepsilon_j$, Theorem \ref{thm:single_prot} and basic composition imply that the recursive protocol is $(\sum_j \varepsilon_j, \sum_j \delta_j)$-differentially private ($\epsilon\leq m$ implies $\epsilon_j\leq 1$). The condition $14\log(2/\delta_j) \geq 27\epsilon_j$ follows immediately from $\log(1/\delta)\geq 2\epsilon$. As $\sum \delta_j=\delta$ and $\sum \epsilon_j=\epsilon$ this protocol is as private as required.

  To explain where the choices of $p_j$ and $\epsilon_j$ come from we provide the analysis to choose them as a part of the accuracy proof. Neglecting $\log(1/\delta)$ factors and using the accuracy estimates of Theorem \ref{thm:single_prot} combined with Lemma~\ref{lemma:bound-outer-mse}, we get:
  \begin{equation*}
    \mse(\mech{P}^{\mathrm{rec}}_{\vec{\gamma},\vec{p},n})=O_\delta\left(\frac{n}{q_m^2}+\sum_{j=1}^m\frac{p_j}{q_{j-1}^2\epsilon_j^2}\right) \enspace.
  \end{equation*}
  We will minimize this expression by optimizing the $p_j$ over the real numbers for simplicity, then rounding to integers turns out not to affect the asymptotics. Differentiating with respect to $p_l$, for $l\in[m]$, gives
  \begin{equation*}
    -\frac{2nq_l^2}{q_{l-1}^2p_l^3q_m^2}+\frac{1}{q_{l-1}^2\epsilon_l^2} \enspace.
  \end{equation*}
  Setting this equal to zero and re-arranging gives
  \begin{equation*}
    p_l^3=\frac{2n\epsilon_l^2q_l^2}{q_m^2} \enspace.
  \end{equation*}
  Dividing this result for two consecutive values of $l$ gives
  \begin{equation*}
    \frac{p_l^3}{p_{l-1}^3}=\frac{\epsilon_l^2p_l^2}{\epsilon_{l-1}^2}
  \end{equation*}
  and thus
  \begin{equation}
    \label{eq:first-ratio}
    p_l=\frac{\epsilon_l^2p_{l-1}^3}{\epsilon_{l-1}^2} \enspace.
  \end{equation}
  Finding an $\vec{\epsilon}$, for fixed $\vec{p}$, that minimizes the above expression for the MSE gives
  \begin{equation*}
    \epsilon_l^3\propto \frac{p_l}{q_{l-1}^2}
  \end{equation*}
  and thus
  \begin{equation*}
    \frac{\epsilon_l^3}{\epsilon_{l-1}^3}=\frac{p_l}{p_{l-1}^3} \enspace.
  \end{equation*}
  Combining this with Equation \ref{eq:first-ratio} we get that $\epsilon_l=\epsilon_{l-1}$, and
  substituting into Equation \ref{eq:first-ratio} gives $  p_l=p_{l-1}^3$.
  We thus take $\epsilon_l=\epsilon/m$ as stated in the theorem and write $p_l$ as $a^{3^l}$.
  Note that $\sum_{i=1}^{k-1} 3^i = \frac{3^k-3}{2}$. The $\mse$ bound now becomes
  \begin{equation*}
    O_\delta\left(na^{3-3^{m+1}}+\frac{a^3m^3}{\epsilon^2}\right) \enspace.
  \end{equation*}
  Choosing $a$ to minimize this gives $a=\left(\frac{(3^m-1)n\epsilon^2}{m^3}\right)^{3^{-m-1}}\approx n^{3^{-m-1}}$ giving a bound on the $\mse$ of
  \begin{equation*}
    O_\delta\left(n^{3^{-m}}\left(1+\frac{m^3}{\epsilon^2}\right)\right) \enspace.
  \end{equation*}
\end{proof}

\begin{remark}  
  We note that if we allow the constant $a$ in the proof to take any real value the upper bound is still of order $(\log \log n)^3$. So better parameter choices cannot get us all the way to constant error with this analysis. Using advanced composition in place of basic still fails to achieve a constant upper bound. We therefore believe this algorithm incurs super-constant error for any privacy preserving choice of parameters.
\end{remark}

\section{Private Summation from Secure Summation}
\label{sec:secure_to_private}
\begin{proof}[Proof of Lemma~\ref{lemma:privatefromsecure}]
Let $\Pi = (\mech{R}_\Pi, \mech{A}_\Pi)$ be the secure exact summation protocol. We will exhibit the resulting protocol $\mech{P} = (\mech{R}, \mech{A})$, with $\mech{R} = \mech{R}_\Pi\circ \tilde{\mech{R}}$ and $\mech{A} = \tilde{\mech{A}}\circ \mech{A}_\Pi$, where $\tilde{\mech{R}}$ and $\tilde{\mech{A}}$ are defined as follows. $\mech{P}$ executes $\Pi$ with $q = \lceil 2n^{3/2} \rceil$, and thus $\tilde{\mech{R}}: [0,1] \mapsto \mathbb{Z}_{\lceil 2n^{3/2} \rceil}$. $\tilde{\mech{R}}(x_i)$ is the result of first computing a randomized fixed-point encoding of the input $x$ with precision $p=\sqrt{n}$, then adding noise $\rv{Z}_i \sim \polya(1/n,e^{-\epsilon/p})-\polya(1/n,e^{-\epsilon/p})$ in $\mathbb{Z}$ and taking the result modulo $q$. $\tilde{\mech{A}}$ decodes $z$ by returning $(z - q)/p$ if $z > \frac{3np}{2}$, and $z/p$ otherwise. This addresses potential over/under-flows of the sum in $\mathbb{Z}_q$.

To show that this protocol is private we will compare the view $\mech{V}_\mech{P}$ to another mechanism $\mech{M}_C$ (which can be considered to be computed in the curator model) which is $\epsilon$-DP and such that $$\SD(\mech{V}_\mech{P}(\vec{x}),\mech{M}_C(\vec{x})) \leq 2^{-\sigma}$$ for all $\vec{x}$, from which the result follows by Lemma~\ref{lemma:SD2delta}.

$\mech{M}_C(\vec{x})$ is defined to be the result of the following procedure. First apply $\tilde{\mech{R}}$ to each input $x_i$, then take the sum $s=\sum_{i=1}^n \tilde{\mech{R}}(x_i)$ and output $\mech{V}_{\Pi}(s,0,\ldots,0)$, the view of the aggregator in the protocol $\Pi$ with first input $s$ and all other inputs $0$.

Note that $s = \sum_{i=1}^n \fp(x_i, p) + \discretelaplace(\alpha)$ with $\alpha = e^{-\epsilon/p}$, where we define the randomized rounding operation $\fp(x, p) = \lfloor xp \rfloor + \bernoulli(xp - \lfloor xp \rfloor)$.
The worst-case sensitivity of $\sum_{i=1}^n \fp(x_i, p)$ under the change of one $x_i$ is $p$. It follows that $s$ is $\epsilon$-DP and thus by the post processing property so is $\mech{M}_C$.

It remains to show that $\SD(\mech{V}_\mech{P}(\vec{x}),\mech{M}_C(\vec{x})) \leq 2^{-\sigma}$, which we will do by demonstrating the existence of a coupling. First let the noise added to input $x_i$ by $\tilde{\mech{R}}$ be the same in both mechanisms and note that this results in the inputs to the randomizers from $\Pi$ used in $\mech{P}$ and $\mech{M}_C$ to have the same sum. It then follows immediately from the worst-case statistical security assumption that the executions of the randomizers and shufflers in $\mech{V}_\mech{P}$ and $\mech{M}_C$ can be coupled to have identical outputs except with probability $2^{-\sigma}$, as required.

Next we show that $\mech{M}_{\mech{P}}$ has MSE $O_{\epsilon}(1)$.
The mean squared error incurred by randomized rounding is bounded by $n/(4p^2)$ by Lemma~\ref{lem:roundingerror}.
The discrete Laplace distribution $\discretelaplace(\alpha)$ has mean zero and variance $2\alpha/(1-\alpha)^2$, so rescaling by $p$ gives a mean squared error of $2\alpha/p^2(1-\alpha)^2$. Summing these two terms would give the exact mean squared error if our arithmetic was in $\mathbb{Z}$, however we are in $\mathbb{Z}_q$ and so need to account for the possibility of under/overflowing. The worst-case mean squared error is bounded by $(q/p)^2$, and the probability of under/overflow bounded by $\alpha^{\frac{q-np}{2}}$ so the following expression is a bound on the mean squared error:
\begin{align*}
  &\frac{2\alpha}{p^2(1-\alpha)^2}+\frac{n}{4p^2}+(q/p)^2\alpha^{\frac{q-np}{2}} \\
  \leq&\frac{2e^{-\epsilon/\sqrt{n}}}{n(1-e^{-\epsilon/\sqrt{n}})^2}+\frac{1}{4}+5n^2e^{-\frac{\epsilon n}{2}} \\
  \leq&\frac{2}{\epsilon^2}+\frac{1}{4}+5n^2e^{-\frac{\epsilon n}{2}} = O(1/\epsilon^2) \enspace.
\end{align*}
Note that the second and third terms can be made arbitrarily small at the expense of more communication by increasing $p$ and $q$.

\end{proof}

The choice $p=\sqrt{n}$ was made so that the error in the discretization was the same order as the error due to the noise added, and this recovers the same order MSE as the curator model. Taking $p=\omega(\sqrt{n})$ results in the leading term of the total MSE still matching the curator model at the cost of a small constant factor increase to communication. 
\section{Proof of Secure Summation}
\label{sec:the_proof}

In this section we give the lemmas required to complete the proof of Theorem~\ref{thm:avg-security} in Section~\ref{sec:proof_outline}. 

\subsection{Reduction to a single input and shuffling step}
\label{sec:single_shuffling}

To analyze the average-case statistical security of $\mech{V}$ we start by upper bounding the expected total variation distance between the outputs of two executions with random inputs by a function of single random input.

\begin{lemma}
Let $\mech{V}_{m,n}$ and $\mech{V}_{m,n}'$ denote two independent executions of the $m$-parallel IKOS protocol. Then we have:
\begin{align*}
&\Ex_{\tup{\rv{X}},\tup{\rv{X}}'} [\TV_{|\tup{\rv{X}},\tup{\rv{X}}'}(\mech{V}_{m,n}(\tup{\rv{X}}), \mech{V}_{m,n}(\tup{\rv{X}}'))] \\
\leq&
\sqrt{q^{mn-1} \Pr[\mech{V}_{m,n}(\tup{\rv{X}}) = \mech{V}_{m,n}'(\tup{\rv{X}})] - 1} \enspace.
\end{align*}
\end{lemma}
\begin{proof}
  We first use a triangle inequality to reduce to proving a bound on the total variation distance from each of $\tup{\rv{X}}$ and $\tup{\rv{X}}'$ to a third variable. That third variable is given by $\tup{\rv{V}}$ which is defined to be uniformly distributed over the tuples in $\mathbb{G}^{mn}$, independently of  $\tup{\rv{X}}$ and $\tup{\rv{X}}'$ except that it shares the same sum.
\begin{align*}
&\Ex_{\tup{\rv{X}},\tup{\rv{X}}'} [\TV_{|\tup{\rv{X}},\tup{\rv{X}}'}(\mech{V}(\tup{\rv{X}}), \mech{V}(\tup{\rv{X}}'))]\\
\leq&
\Ex_{\tup{\rv{X}},\tup{\rv{X}}'} [\TV_{|\tup{\rv{X}},\tup{\rv{X}}'}(\mech{V}(\tup{\rv{X}}), \tup{\rv{V}})) + \TV_{|\tup{\rv{X}},\tup{\rv{X}}'}(\tup{\rv{V}}, \mech{V}(\tup{\rv{X}}'))]
\\
=&
\Ex_{\tup{\rv{X}}} [\TV_{|\tup{\rv{X}}}(\mech{V}(\tup{\rv{X}}), \tup{\rv{V}})]
+
\Ex_{\tup{\rv{X}'}} [\TV_{|\tup{\rv{X}}'}(\tup{\rv{V}}, \mech{V}(\tup{\rv{X}}'))]
\\
=&
2 \Ex_{\tup{\rv{X}}} [\TV_{|\tup{\rv{X}}}(\mech{V}(\tup{\rv{X}}), \tup{\rv{V}})] \enspace.
\end{align*} 

We now expand the total variation and write it as an expectation over $\tup{\rv{V}}$ as follows:
\begin{align*}
2 \TV_{|\tup{\rv{X}}}(\mech{V}(\tup{\rv{X}}), \tup{\rv{V}})
&=
\sum_{\tup{v} \in \mathbb{G}^{mn}} | \Pr_{|\tup{\rv{X}}}[\mech{V}(\tup{\rv{X}}) = \tup{v}] - \Pr[\tup{\rv{V}} = \tup{v}]|
\\
&=
\sum_{\tup{v} \in \mathbb{G}^{mn} : \sum \tup{v} = \sum \tup{\rv{X}}} | \Pr_{|\tup{\rv{X}}}[\mech{V}(\tup{\rv{X}}) = \tup{v}] - q^{1-mn}|
\\
&=
q^{mn-1} \Ex_{\tup{\rv{V}}} [|\Pr_{|\tup{\rv{X}},\tup{\rv{V}}}[\mech{V}(\tup{\rv{X}}) = \tup{\rv{V}}] - q^{1-mn}|] \enspace.
\end{align*}

The final task is to bound the remaining expectation.
We start by defining the random variable $\rv{Z} = \rv{Z}(\trv{X},\trv{V}) := \Pr_{|\tup{\rv{X}},\tup{\rv{V}}}[\mech{V}(\tup{\rv{X}}) = \tup{\rv{V}}]$.
Note that because both $\mech{V}(\tup{\rv{X}})$ and $\tup{\rv{V}}$ follow the same uniform distribution over tuples in $\mathbb{G}^{mn}$ conditioned to having the same sum, we have
\begin{align*}
\Ex_{\trv{X},\trv{V}}[\rv{Z}] = \Pr[\mech{V}(\tup{\rv{X}}) = \tup{\rv{V}}] = q^{1-mn} \enspace.
\end{align*}
Therefore, the expectation that we need to bound takes the simple form $\Ex[|\rv{Z} - \Ex[\rv{Z}]|]$, and can be bounded in terms of $\Ex[\rv{Z}^2]$ via Jensen's inequality:
\begin{align*}
\Ex[|\rv{Z} - \Ex[\rv{Z}]|]
\leq
\sqrt{\Var[\rv{Z}]}
=
\sqrt{\Ex[\rv{Z}^2] - \Ex[\rv{Z}]^2} \enspace.
\end{align*}

Now recall that if $\rv{A}, \rv{A}' \in A$ are i.i.d. random variables, then we have
\begin{align*}
\Pr[\rv{A} = \rv{A}'] = \sum_{a \in A} \Pr[\rv{A} = a]^2 \enspace.
\end{align*}
Using this identity we can write the expectation of $\rv{Z}^2$ over the randomness in $\trv{V}$ in terms of the probability that two independent executions of $\mech{V}(\trv{X})$ (conditioned on $\trv{X})$ yield the same result:
\begin{align*}
\Ex_{\trv{V}}[\rv{Z}^2]
&=
q^{1-mn} \sum_{\tup{v} \in \mathbb{G}^{mn} : \sum \tup{v} = \sum \tup{\rv{X}}}
\Pr_{|\tup{\rv{X}}}[\mech{V}(\tup{\rv{X}}) = \tup{v}]^2
\\
&=
q^{1-mn} \Pr_{|\tup{\rv{X}}}[\mech{V}(\tup{\rv{X}}) = \mech{V}'(\tup{\rv{X}})]
\enspace.
\end{align*}
Putting the pieces together completes the proof:
\begin{align*}
&\Ex_{\tup{\rv{X}},\tup{\rv{X}}'} [\TV_{|\tup{\rv{X}},\tup{\rv{X}}'}(\mech{V}(\tup{\rv{X}}), \mech{V}(\tup{\rv{X}}'))]\\
\leq&
2 \Ex_{\tup{\rv{X}}} [\TV_{|\tup{\rv{X}}}(\mech{V}(\tup{\rv{X}}), \trv{V})] \\
\leq&
q^{mn-1} \Ex_{\tup{\rv{X}}, \tup{\rv{V}}} [|\Pr_{|\tup{\rv{X}},\tup{\rv{V}}}[\mech{V}(\tup{\rv{X}}) = \tup{\rv{V}}] - q^{1-mn}|] \\
\leq&
\sqrt{q^{mn-1} \Pr[\mech{V}(\tup{\rv{X}}) = \mech{V}'(\tup{\rv{X}})] - 1} \enspace.
\end{align*}
\end{proof}

To further simplify the bound in previous lemma we can write the probability $\Pr[\mech{V}_{m,n}(\tup{\rv{X}}) = \mech{V}_{m,n}'(\tup{\rv{X}})]$ in terms of a single permutation step.
For that purpose we introduce the notation $\mech{V}_{m,n} = \mech{S}_{m,n} \circ \tup{\mech{R}}_{m,n}$, where:
\begin{itemize}
\item $\tup{\mech{R}}_{m,n} : \mathbb{G}^n \to \mathbb{G}^{n m}$ is the randomized map that given $\tup{x} = (x_1,\ldots,x_n)$ generates the shares $(\rv{Y}_i^{(1)}, \ldots, \rv{Y}_i^{(m)}) = \mech{R}(x_i)$ and arranges them in order first by share id and then by user:
\begin{align*}
\tup{\mech{R}}_{m,n}(\tup{x}) = (\rv{Y}_1^{(1)}, \ldots, \rv{Y}_n^{(1)}, \ldots, \rv{Y}_1^{(m)}, \ldots, \rv{Y}_n^{(m)}) \enspace.
\end{align*}
\item $\mech{S}_{m,n} : \mathbb{G}^{nm} \to \mathbb{G}^{nm}$ is a random permutation of its inputs obtained by applying $m$ independent shufflers $\mech{S}^{(j)}$, $j \in [m]$, to the inputs in blocks of $n$:
\begin{align*}
&\mech{S}_{m,n}(y_1^{(1)}, \ldots, y_n^{(1)}, \ldots, y_1^{(m)}, \ldots, y_n^{(m)})\\
=& (\mech{S}^{(1)}(y_1^{(1)}, \ldots, y_n^{(1)}) \cdots \mech{S}^{(m)}(y_1^{(m)}, \ldots, y_n^{(m)}))
\end{align*}
\end{itemize}
It is important to note that $\mech{S}_{m,n}$ produces random permutations of $[mn]$ which are uniformly distributed in the subgroup of all permutations which arise as the parallel composition of $m$ uniform permutations on $[n]$.
Equipped with these observations, it is straightforward to verify the following identity.

\begin{lemma}
Let $\tup{\mech{R}}_{m,n}$ and $\tup{\mech{R}}_{m,n}'$ denote two independent executions of the additive sharing step in $\mech{V}_{m,n}(\tup{\rv{X}}) = \mech{S}_{m,n} \circ \tup{\mech{R}}_{m,n}$.
Then we have
\begin{align*}
\Pr[\mech{V}_{m,n}(\tup{\rv{X}}) = \mech{V}'_{m,n}(\tup{\rv{X}})]
&=
\Pr[\tup{\mech{R}}_{m,n} (\tup{\rv{X}}) = \mech{S}_{m,n} \circ \tup{\mech{R}}_{m,n}' (\tup{\rv{X}})] \enspace.
\end{align*}
\end{lemma}
\begin{proof}
We drop all subscripts for convenience.
The result follows directly from the fact that $\mech{S}$ is uniform over a subgroup of permutations, which implies that the inverse of $\mech{S}$ and the composition of two independent copies of $\mech{S}$ both follow the same distribution as $\mech{S}$.
Thus, we can write:
\begin{align*}
\Pr[\mech{V}(\tup{\rv{X}}) = \mech{V}'(\tup{\rv{X}})]
&=
\Pr[\mech{S} \circ \tup{\mech{R}} (\tup{\rv{X}}) = \mech{S}' \circ \tup{\mech{R}}' (\tup{\rv{X}})]
\\
&=
\Pr[\tup{\mech{R}} (\tup{\rv{X}}) = \mech{S}^{-1} \circ \mech{S}' \circ \tup{\mech{R}}' (\tup{\rv{X}})]
\\
&=
\Pr[\tup{\mech{R}} (\tup{\rv{X}}) = \mech{S} \circ \tup{\mech{R}}' (\tup{\rv{X}})] \enspace.
\end{align*}
\end{proof}

Putting these two lemmas together yields the following bound.

\begin{lemma}\label{lem:single_shuffling}
Let $\mech{V}_{m,n}$ and $\mech{V}_{m,n}'$ denote two independent executions of the $m$-parallel IKOS protocol. Then we have:
\begin{align*}
&\Ex_{\tup{\rv{X}},\tup{\rv{X}}'} [\TV_{|\tup{\rv{X}},\tup{\rv{X}}'}(\mech{V}_{m,n}(\tup{\rv{X}}), \mech{V}_{m,n}(\tup{\rv{X}}'))]\\
\leq&
\sqrt{q^{mn-1} \Pr[\tup{\mech{R}}_{m,n} (\tup{\rv{X}}) = \mech{S}_{m,n} \circ \tup{\mech{R}}_{m,n}' (\tup{\rv{X}})] - 1} \enspace.
\end{align*}  
\end{lemma}

\subsection{Reduction to a problem on random graphs}
\label{sec:graph_argument}

\begin{definition}
  A random $n$-vertex $2m$-regular multigraph $G$ is from the \emph{unconditioned permutation model}, denoted $R^*(n,2m)$, if it can be constructed as follows. Start with $n$ vertices and no edges. Take a set of $m$ uniformly random and independent permutations of the vertices, denoted $\{\pi_i\}_{i=1}^m$.  For each vertex $v$ and each index $i\in [m]$, add an edge between $v$ and $\pi_i(v)$. We say that $G$ is generated by $\{\pi_i\}_{i=1}^m$.
\end{definition}

Note that $G$ may have self-loops. Let $C(G)$ be number of connected components of a graph $G$.

\begin{lemma}
  \label{lem:graph_argument}
  Let $G$ be drawn from $R^*(n,2m)$, then
  \begin{equation*}
  \Pr[\tup{\mech{R}} (\tup{\rv{X}}) = \mech{S} \circ \tup{\mech{R}}' (\tup{\rv{X}})]\leq \Ex[q^{C(G)-mn}]
  \end{equation*}
\end{lemma}

\begin{proof}
  Note that, by the tower law,
  \begin{align*}
    \Pr[\tup{\mech{R}} (\tup{\rv{X}}) = \mech{S} \circ \tup{\mech{R}}' (\tup{\rv{X}})]=\Ex[\Pr[\tup{\mech{R}} (\tup{\rv{X}}) = \mech{S} \circ \tup{\mech{R}}' (\tup{\rv{X}})|\mech{S}]].
  \end{align*}
  Let $G_{\mech{S}}$ be the unconditioned permutation model graph, with vertex set $[n]$, generated by the $m$ permutations used in $\mech{S}$. Note that, it suffices to show that
  \begin{equation*}
    \Pr[\tup{\mech{R}} (\tup{\rv{X}}) = \mech{S} \circ \tup{\mech{R}}' (\tup{\rv{X}})|\mech{S}]=q^{C(G_{\mech{S}})-mn}.
  \end{equation*}

  For notational convenience, we will define a deterministic reordering of $\tup{\mech{R}} (\tup{\rv{X}})$ and  $\mech{S} \circ \tup{\mech{R}}' (\tup{\rv{X}})$ as follows. Consider the permutation $P:[mn]\rightarrow [mn]$
  \begin{equation*}
    P(j)=\left\lfloor \frac{j-1}{m}\right\rfloor+n(j-1 \textrm{ mod } m)+1.
  \end{equation*}

  Define $\rv{U},\rv{U}'\in \mathbb{G}^{mn}$ by $\rv{U}_j=\tup{\mech{R}} (\tup{\rv{X}})_{P(j)}$ and $\rv{U}'_j=P\circ \mech{S} \circ \tup{\mech{R}}' (\tup{\rv{X}})_{P(j)}$. Note that $P$ is such that the shares from each input are grouped together (in order) in $\rv{U}$. Consequently, $\rv{U}'$ groups together collections of $m$ shares, one from the output of each shuffler. Thus it suffices to show that
  \begin{equation*}
    \Pr[ \rv{U}=\rv{U}' |\mech{S}]=q^{C(G_{\mech{S}})-mn}.
  \end{equation*}
  
  For $j\in [mn]$, let $A_j$ be the event that $\rv{U}_j=\rv{U}'_j$. Now define $p_j:=\Pr[A_j|A_1,...,A_{j-1},\mech{S}]$, thus
  \begin{equation*}
    \Pr[ \rv{U}=\rv{U}' |\mech{S}]=\prod_{j=1}^{mn} p_j.
  \end{equation*}

  First we consider values of $j$ that are not divisible by $m$, i.e. they are not the final share in a group of $m$. For such a $j$, we claim $p_j=q^{-1}$. To see this, condition on $\tup{\rv{X}}$ and $\tup{\mech{R}}'$, in addition to $A_1,...,A_{j-1}$. Note that $\rv{U}_j$ and $\rv{U}_{j+1}$ only depend upon anything we've conditioned on via their sum. Therefore $\rv{U}_j$ is still uniformly distributed and has probability $q^{-1}$ of being equal to $\rv{U}'_j$.

For an index $i\in [mn]$ we define the vertex corresponding to $i$ to be the vertex $\lceil i/m \rceil$, and we define $C_i$ to be the set of vertices in the same connected component as this vertex in $G_{\mech{S}}$. For the remaining $j$'s, we distinguish the case where the corresponding vertex is the highest index in $C_j$ and the case where it isn't.

In the first case,
\begin{equation*}
  \sum_{i \textrm{ s.t. } C_i=C_j}\rv{U}'_i=\sum_{i \textrm{ s.t. } C_i=C_j}\tup{\mech{R}}' (\tup{\rv{X}})_{P(i)}
\end{equation*}
as the sums have the same summands in a different order. Further,
\begin{equation*}
  \sum_{i \textrm{ s.t. } C_i=C_j}\tup{\mech{R}}' (\tup{\rv{X}})_{P(i)}=\sum_{i \textrm{ s.t. } C_i=C_j}\rv{U}_i
\end{equation*}
as they both represent sharings of the same input values and
\begin{equation*}
  \sum_{\substack{i \textrm{ s.t. } C_i=C_j \\ i\neq j}}\rv{U}'_i=\sum_{\substack{i \textrm{ s.t. } C_i=C_j \\ i\neq j}}\rv{U}_i
\end{equation*}
as we are conditioning on $A_1,...,A_{j-1}$. Putting these together we can conclude that $p_j=1$.

For the second case, we will find that $p_j=q^{-1}$. We will show this by showing that if we condition on the value of $\rv{U}_j$ then $\rv{U}'_j$ is still uniformly distributed. That is to say that the number of possible outcomes fitting those conditions with each value of $\rv{U}_j$ is independent of that value. To show that these sets of outcomes have the same size we will partition the possible outcomes into sets of size $q$, with $\rv{U}'_j$ taking each value in $\mathbb{G}$ exactly once in each set. This will be possible because the structure of $G_{\mech{S}}$ allows us to change the value of $\rv{U}'_j$ and other values to preserve what is being conditioned on in an algebraically principled way. If $\mathbb{G}=\mathbb{Z}_q$, for some prime $q$, i.e. the set of possible outcomes forms a vector space, this can be thought of as follows. The space of possible outcomes consitent with the conditions is a subspace of the space of all outcomes. Thus showing that this subspace contains two possible values for $\rv{U}'_j$ suffices by the nice algebraic properties of vector spaces. That there is more than one possible value of $\rv{U}'_j$ is a consequence of $G_{\mech{S}}$ ``connecting the $j$th share to later shares''. The following paragraphs make this formal in the more general setting of any abelian group $\mathbb{G}$.

Consider the set $\mathcal{T}$ of choices of $(\rv{U}\cdot \rv{U}')\in \mathbb{G}^{2mn}$ that are consistent with $A_1,...,A_{j-1}$ (and a value of $\tup{\rv{X}}$). We consider the group action of $\mathbb{G}^{2mn}$ on itself by addition. We will show that, there exists a homomorphism $\mathbb{G}\rightarrow \mathbb{G}^{2mn}$ mapping $g$ to $u_g$ with the following property.  The action of $u_g$ on $\mathbb{G}^{2mn}$ fixes $\mathcal{T}$ and $\rv{U}_j$ and adds $g$ to $\rv{U}'_j$. Therefore, the equivalence relation, of being equal upto adding $u_g$ for some $g$, partitions $\mathcal{T}$ into subsets of size $q$ each containing one value for which $A_j$ holds. It follows, from the fact that each entry in $\mathcal{T}$ is equally likely, that $p_j=q^{-1}$.

To find such a homomorphism, note that there is a path in $G_S$ from the vertex corresponding to the $j$th share to a higher index vertex. This is equivalent to saying that there is a sequence $(a_1,b_1,a_2,b_2,...,a_l,b_l,a_{l+1})$ with the following properties. The $a_i$ and $b_i$ are elements of $[mn]$ and should be interpreted as indexes of $\mathbb{G}^{mn}$. For all $i\in [l]$, $\pi(b_i)=a_i$ and $b_i$ and $a_{i+1}$ correspond to the same vertex. We have $a_1=j$, $b_l>j$, $a_i\neq a_{i'}$ for any $i\neq i'$, and $b_i< j$ for all $i<l$. Let $u_g$ be the element of $\mathbb{G}^{2mn}$ with a $g$ in entries $a_2,...,a_{l+1},b_1+mn,...,b_{l}+mn$ and the identity everywhere else.

Adding $u_g$ doesn't change the truth of $A_1,...,A_{j-1}$ because $\rv{U}_{a_i}$ and $\rv{U}'_{a_i}=\tup{\mech{R}}'(\tup{\rv{X}})_{b_i}$ are always incremented together, with the exception of when $i=1$ or $l+1$ which is fine because then $a_i\geq j$. In the case of $i=1$ this adds $g$ to $\rv{U}'_j$ without changing $\rv{U}_j$. The consistency of the implied values of $\tup{\rv{X}}$ is maintained because $\rv{U}_{a_i}$ and $\tup{\mech{R}}'(\tup{\rv{X}})_{P(b_{i-1})}$ are always incremented together and affect the $\tup{\rv{X}}$ implied by $\rv{U}$ the same as that implied by $\tup{\mech{R}}'(\tup{\rv{X}})$. Thus, this $u_g$ has the properties we claimed and $p_j=q^{-1}$.

Tying this together we have that
\begin{align*}
\Pr[\rv{U} = \rv{U}']&=\Ex[\Pr[\rv{U}=\rv{U}'|\mech{S}]] \\
                                                                                  &=\Ex[\prod_{j=1}^{mn} p_j] \\
                                                                                  &=\Ex[q^{C(G_S)-mn}] \\
                                                                                  &=\Ex[q^{C(G)-mn}]\enspace.
\end{align*}
\end{proof}
 
\subsection{Understanding the number of connected components of \texorpdfstring{$G$}{G}}
\label{sec:bound_on_graph}

\begin{lemma}
  \label{lem:bound_on_graph}
   Let $n\geq 19$, $m\geq 3$ and $q\leq \frac{1}{2}\left(\frac{n}{e}\right)^{m-1}$. Let $G$ be drawn from $R^*(n,2m)$, then
  \begin{equation*}
    \PP(C(G)=c)\leq \frac{1.5^{c-1}}{c!}\left(\frac{e}{n}\right)^{(m-1)(c-1)}.
  \end{equation*}
  Therefore,
  \begin{align*}
    \Ex[q^{C(G)}]&\leq q + q^2\left(\frac{n}{e}\right)^{1-m} \enspace.
  \end{align*}
\end{lemma}

\begin{proof}
  Let $p(n,c)=\Pr[C(G^n)=c]$, we will show that the bound in the theorem holds by induction on $c$. For $c=1$ the bound is trivial as the right hand side equals $1$. For $c>1$,
  \begin{align*}
    p(n,c)&=\frac{1}{c}\sum_{S\subset [n]}\Pr[\textrm{No edge from $S$ to $[n]-S$}]p(|S|,1)p(n-|S|,c-1) \\
          &=\frac{1}{c}\sum_{s=1}^{n-c+1}\binom{n}{s}\binom{n}{s}^{-m}p(s,1)p(n-s,c-1).
  \end{align*}
  We now bound this expression, using the induction hypothesis, to find that
  \begin{align}
    \label{eq:sum_to_bound}
    p(n,c)&\leq \frac{1}{c}\sum_{s=1}^{n-c+1}\binom{n}{s}^{1-m}\frac{1.5^{c-2}}{(c-1)!}\left(\frac{e}{n-s}\right)^{(m-1)(c-2)} \\
          &= \frac{1.5^{c-1}}{c!}\left(\frac{e}{n}\right)^{(m-1)(c-1)}\frac{2e^{1-m}}{3}\sum_{s=1}^{n-c+1}\left(\frac{(n-s)!s!n^{c-1}}{n!(n-s)^{c-2}}\right)^{m-1}.
  \end{align}
  To complete the proof it suffices to show that this sum on the right is at most $1.5e^{m-1}$. Call the $s$th summand from this sum $a_s$. We separate the summands into three cases, depending on whether $s$ is greater than $n/10$ and/or less than $3n/4$. Firsty, if $s\leq n/10$, then
  \begin{align*}
    \frac{a_{s}}{a_{s-1}}&=\left(\frac{s}{n-s+1}\left(\frac{n-s+1}{n-s}\right)^{c-2}\right)^{m-1} \\
                         &\leq \left(\frac{s}{n-s}e^{\frac{c-2}{n-s}}\right)^{m-1} \\
                         &\leq \left(\frac{e^{\frac{10}{9}}}{9}\right)^{2} \\
                         &\leq \frac{1}{8}.
  \end{align*}
  Thus we can bound the early summands with a geometric series as follows.
  \begin{align*}
    \sum_{s=1}^{\lfloor n/10 \rfloor}a_s&\leq \sum_{s=1}^{\lfloor (n-c)/10 \rfloor}\frac{a_1}{8^{s-1}} \\
                                            &\leq \sum_{s=1}^\infty \frac{a_1}{8^{s-1}} \\
                                            &\leq \frac{8a_1}{7} \\
                                            &=\frac{8}{7}\left(\frac{n}{n-1}\right)^{(c-2)(m-1)} \\
                                            &\leq \frac{8}{7}e^{\frac{(c-1)(m-1)}{n}}\leq \frac{8}{7}e^{(m-1)}
  \end{align*}

  We now similarly consider the terms with $s\geq 3n/4$. For these values of $s$,
  \begin{align*}
    \frac{a_{s+1}}{a_{s}}&=\left(\frac{s+1}{n-s}\left(\frac{n-s}{n-s-1}\right)^{c-2}\right)^{m-1} \\
                         &\geq\left(\frac{s}{n-s}\right)^{m-2} \\
                         &\geq 9
  \end{align*}
  If $c>n/4$ then there are no summands for $s\geq 3n/4$. Otherwise we can bound the late summands with a geometric series as follows.
  \begin{align*}
    \sum_{s=\lceil 3n/4 \rceil}^{n-c+1}a_s&\leq \sum_{s=\lceil 3n/4 \rceil}^{n-c+1}\frac{a_{n-c+1}}{9^{n-c+1-s}} \\
                                            &\leq \sum_{s=-\infty}^{n-c+1} \frac{a_{n-c+1}}{9^{n-c+1-s}} \\
                                            &= \frac{9a_{n-c+1}}{8} \\
                                            &=\frac{9}{8}\left(\frac{(c-1)!n^{c-1}(n-c+1)!}{n!(c-1)^{c-2}}\right)^{m-1}.
  \end{align*}
  Applying Sterling's bound, $\sqrt{2\pi}n^{n+\frac{1}{2}}e^{-n}\leq n!\leq e n^{n+\frac{1}{2}}e^{-n}$, to the factorials in the above expression bounds it above by,
  \begin{equation*}
    \frac{9}{8}\left(\frac{e^2}{\sqrt{2\pi}}(c-1)^{1.5}\left(1-\frac{c-1}{n}\right)^{n-c+1.5}\right)^{m-1}.
  \end{equation*}
  As $n\geq 19$ and $c\leq n/4$, this is maximised for $c=3$, and as we also have $m\geq 3$ this results in the bound
  \begin{equation*}
    \frac{9}{8}\left(\frac{e^2}{\sqrt{2\pi}}2\sqrt{2}(1-\frac{2}{n})^{n-1.5}\right)^{m-1}\leq \left(1.27\right)^{m-1}.
  \end{equation*}
  
  Finally we consider the case of $n/10<s<3n/4$. Let $\alpha=s/n$. Substituting this into $a_s$ gives
  \begin{equation*}
    \left(\frac{((1-\alpha)n)!(\alpha n)!}{(n-1)!(1-\alpha)^{c-2}}\right)^{m-1}.
  \end{equation*}
  Applying Sterling's bound again bounds this expression by 
  \begin{align*}
    \left(\frac{e^2}{\sqrt{2\pi}}\sqrt{n}(1-\alpha)^{2.5-c+(1-\alpha)n}\alpha^{\alpha n +\frac{1}{2}}\right)^{m-1}&\leq \left(\frac{e^2\sqrt{n}}{\sqrt{2\pi}}\alpha^{\alpha n}\right)^{m-1}\enspace.
  \end{align*}
  Where the inequality holds because $(1-\alpha)\leq 1$ and, for any summand that appears in the sum, $2.5-c+(1-\alpha) n>0$.
  The final expression is maximised for $\alpha=3/4$ and there are fewer than $3n/5$ summands with $n/10<s<3n/4$. Therefore the sum of all of these terms can be bounded by,
  \begin{align*}
    \frac{3n}{5}\left(\frac{e^2\sqrt{n}}{\sqrt{2\pi}}\left(\frac{3}{4}\right)^{\frac{3n}{4}}\right)^{m-1}&\leq \left(en\left(\frac{3}{4}\right)^{\frac{3n}{4}}\right)^{m-1} \\
    &\leq 1.
  \end{align*}
  Where we have used that $m\geq 3$ and $n\geq 19$. Adding these up the sum as a whole is bounded by
  \begin{equation*}
    \frac{8}{7}e^{m-1}+1+\left(1.27\right)^{m-1}<1.5e^{m-1}.
  \end{equation*}
  To conclude the proof we consider the expectation. Below we apply the definition of expectation with the bound on the probability above.
  \begin{equation*}
    \Ex[q^{C(G)}]\leq\sum_{c=1}^{n}q^c\frac{1.5^{c-1}}{c!}\left(\frac{e}{n}\right)^{(m-1)(c-1)}
  \end{equation*}
  Notice that every term after the second is at most $\frac{qe^{m-1}}{2n^{m-1}}$ times the previous term, thus
  \begin{equation*}
    \Ex[q^{C(G)}]\leq q+\frac{3q^2}{4}\left(\frac{n}{e}\right)^{1-m}\sum_{i=0}^{\infty} \left(\frac{qe^{m-1}}{2n^{m-1}}\right)^i
  \end{equation*}
  Then using that $q\leq \frac{1}{2}\left(\frac{n}{e}\right)^{m-1}$ we bound the sum by $4/3$ to find
  \begin{equation*}
    \Ex[q^{C(G)}]\leq q+q^2\left(\frac{n}{e}\right)^{1-m} \enspace.
  \end{equation*}
\end{proof} 
\section{From Average-Case to Worst-Case Security}
\label{sec:avg_to_worst}
\begin{proof}[Proof of Lemma~\ref{lem:avg_to_worst}]
Fix a pair of inputs $\tup{x}$ and $\tup{x}'$ with the same sum.
Since the output of $\mech{V}_{m+1,n}(\tup{x})$ can be simulated directly from the output of $\tilde{\mech{V}}_{m,n}(\tup{x})$ by applying a random permutation to the last $n$ elements, we have $\TV(\mech{V}_{m+1,n}(\tup{x}),\mech{V}_{m+1,n}(\tup{x}')) \leq \TV(\tilde{\mech{V}}_{m,n}(\tup{x}), \tilde{\mech{V}}_{m,n}(\tup{x}'))$, and therefore it suffices to show that $\tilde{\mech{V}}_{m,n}$ provides worst-case statistical security with parameter $\sigma$.

The key observation that allows us to reduce the worst-case security of $\tilde{\mech{V}}_{m,n}$ to the average-case security of $\mech{V}_{m,n}(\tup{x})$ is to observe that the addition of an extra share can be interpreted as adding a random value to each user's input, effectively making the inputs uniformly random.
To formalize this intuition we observe that $\mech{R}_m$ admits a recursive decomposition as follows.
Let $\rv{U} \in \dom{G}$ be a uniformly random group element and $x \in \dom{G}$. Then we have $\mech{R}_1(x) = x$ and, for $m \geq 1$,
\begin{align*}
\mech{R}_{m+1}(x) = (\mech{R}_{m}(x - \rv{U}), \rv{U}) \enspace.
\end{align*}
Expanding this identity into the definition of $\tilde{\mech{V}}$ and writing $\tup{\rv{U}} = (\rv{U}_1, \ldots, \rv{U}_n)$ for the uniform random variables arising from applying the above expression for $\mech{R}_{m+1}$ to the input from each user, we obtain
\begin{align*}
\tilde{\mech{V}}_{m}(\tup{x})
=
(\mech{V}_{m}(\tup{x} - \tup{\rv{U}}), \tup{\rv{U}}) \enspace.
\end{align*}
Note that here $\tup{x} - \tup{\rv{U}}$ is a uniform random vector in $\dom{G}^n$.
The result now follows from matching the uniform randomness from $\tup{\rv{U}}$ observed when executing the protocol with two inputs with the same sum:
\begin{align*}
\TV(\tilde{\mech{V}}_{m+1}(\tup{x}), \tilde{\mech{V}}_{m+1}(\tup{x}'))
&=
\TV((\mech{V}_{m}(\tup{x} - \tup{\rv{U}}), \tup{\rv{U}}), (\mech{V}_{m}(\tup{x}' - \tup{\rv{U}}'), \tup{\rv{U}}'))
\\
&=
\Ex_{\tup{\rv{U}}} [\TV_{|\tup{\rv{U}}}(\mech{V}_{m}(\tup{x} - \tup{\rv{U}}), \mech{V}_{m}(\tup{x}' - \tup{\rv{U}}))]
\\
&=
\Ex_{\tup{\rv{X}},\tup{\rv{X}}'} [\TV_{|\tup{\rv{X}},\tup{\rv{X}}'}(\mech{V}_{m}(\tup{\rv{X}}), \mech{V}_{m}(\tup{\rv{X}}'))]
\enspace,
\end{align*}
where $\tup{\rv{X}}$ is a uniformly random tuple of $n$ group elements and $\tup{\rv{X}}'=\tup{\rv{X}}-\tup{x}+\tup{x}'$ is a constant offset. Thus they are both uniformly distributed and have the same sum. Therefore we can complete the proof by using the assumption of average case security for $\mech{V}_{m,n}$.
\end{proof} 
\section{The Original IKOS Proof}
\label{sec:ikos_proof}

In this section we provide a proof of Lemma~\ref{lemma:IKOS}, all the ideas for the proof are provided in \cite{ikos} but we reproduce the proof here keeping track of constants to facilitate setting parameters of the protocol. The following definition and lemma from \cite{IZ89} are fundamental to why this protocol is secure.

Let $H$ be a family of functions mapping $\{0, 1\}^n$ to
$\{0, 1\}^l$ . We say $H$ is universal or a universal family of
hash functions if, for $h$ selected uniformly at random from $H$, for every $x, y \in \{0, 1\}^n$, $x\neq y$,
\begin{equation*}
  \PP[h(x) = h(y)]=2^{-l}.
\end{equation*}

\begin{lemma}[Leftover Hash Lemma (special case)]
  Let $D\subset \{0, 1\}^n$, $s > 0$, $|D| \geq 2^{l+2s}$ and let $H$ be a universal family of hash functions mapping $n$ bits to $l$ bits. Let $h$, $d$ and $U$ be chosen independently uniformly at random from $H$, $D$ and $\{0,1\}^{l}$ respectively. Then 
\begin{equation*}
\SD\left( (h, h(d)), (h,U) \right)\leq 2^{-s}
\end{equation*}
\end{lemma}

To begin with we consider the case of securely adding two uniformly random inputs $X,Y\in \mathbb{Z}_q$. Recall that $\Pi$ is the protocol of the statement of the lemma, and let $V(x,y)$ be shorthand for $\mech{M}_{\mech{R}_\Pi}((x,y)) = \mech{S}\circ\mech{R}_\Pi((x,y))$, i.e. the view of the analyzer in an execution of protocol $\Pi$ with inputs $x,y$. We write $V$ for $V(X,Y)$ and $V(x)$ for $V(x,Y)$. Finally let $U$ be an independent uniformly random element of $\mathbb{Z}_q$.

\begin{lemma}
  Suppose $\log\binom{2m}{m}\geq \lceil \log(q) \rceil +2s$. Then, $\SD((V,X),(V,U))\leq 2^{-s}$.
\end{lemma}
\begin{proof}
  For $a\in \mathbb{Z}_q^{2m}$ and $\pi \in \binom{[2m]}{m}$ let $h_a(\pi)=\sum_{i\in \pi}a_i$. $(h_a)_{a\in \mathbb{Z}_q^{2m}}$ are a universal family of hash functions from $\binom{[2m]}{m}$ to $\mathbb{Z}_q$. Let $d$ be an independent uniformly random element of $\binom{[2m]}{m}$. Note that $(V,h_V(d))$ has the same distribution as $(V,X)$, which follows from the intuition that $V$ corresponds to $2m$ random numbers shuffled together, and $x$ can be obtained by adding up $m$ of them, and letting $y$ be the sum of the rest.

  The result now follows immediately from the fact that the Leftover Hash Lemma implies that $\SD((V,h_v(d)),(V,U))\leq 2^{-s}$.
\end{proof}

Now we can use this to solve the case of two arbitrary inputs.

\begin{lemma}
  If $x,y,x',y'\in \mathbb{Z}_q$ satisfy $x+y=x'+y'$, then we have
  \begin{equation*}
    \SD(V(x,y),V(x',y'))\leq 2q^2\SD((V,X),(V,U)).
  \end{equation*}
\end{lemma}
\begin{proof}
  Markov's inequality provides that
  \begin{equation*}
    \SD(V(x),V)\leq q\SD((V,X),(V,U)) \hspace{10mm} \forall x\in \mathbb{Z}_q
  \end{equation*}
  and thus by the triangle inequality
  \begin{equation*}
    \SD(V(x),V(x'))\leq 2q\SD((V,X),(V,U)).
  \end{equation*}
  Note that
  \begin{align*}
    \SD(V(x),V(x'))&=\sum_{t\in \mathbb{Z}_q}\SD(V(x)|Y=t-x,V(x')|Y=t-x')/q \\
    &=\frac{1}{q}\sum_{y\in \mathbb{Z}_q} \SD(V(x,y),V(x',y+x-x'))
  \end{align*}
  and so for every $x,y,x'\in \mathbb{Z}_q$ and $y'=y+x-x'$ we have
  \begin{equation*}
    \SD(V(x,y),V(x',y'))\leq q\SD(V(x),V(x')).
  \end{equation*}
  Combining the last two inequalities gives the result.
\end{proof}

Combining these two lemmas gives that, for $x,y,x',y'\in \mathbb{Z}_q$ such that $x+y=x'+y'$,
\begin{align}
  \SD(V(x,y),V(x',y'))&\leq 2q^22^{-\frac{\log\binom{2m}{m}-\lceil \log(q) \rceil}{2}}\nonumber \\
  \label{al:binombound}
                     &\leq 2^{-\frac{m}{2}+1+\frac{5\lceil \log(q) \rceil}{2}}
\end{align}

From which the following lemma is immediate
\begin{lemma}
  \label{lemma:twocase}
  If $m\geq 2+5\lceil\log(q)\rceil+2\sigma$ and $x,y,x',y'\in \mathbb{Z}_q$ such that $x+y=x'+y'$ then 
  \begin{equation*}
    \SD(V(x,y),V(x',y'))\leq 2^{-\sigma}
  \end{equation*}
\end{lemma}

We will now generalize to the case of $n$-party summation.

\begin{proof}[Proof of Lemma \ref{lemma:IKOS}]
Let $\vec{x},\vec{x}'\in \mathbb{Z}_q^n$ be two distinct possible inputs to the protocol, we say that they are related by a basic step if they have the same sum and only differ in two entries. It is evident that any two distinct inputs with the same sum are related by at most $n-1$ basic steps. We will show that if $m$ is taken to be $2+5\lceil\log(q)\rceil+\lceil 2\sigma+ 2\log(n-1) \rceil$ and $\vec{x}$ and $\vec{x}'$ are related by a basic step then
\begin{equation}
  \label{eq:basicbound}
  \SD(V(\vec{x}),V(\vec{x}'))\leq \frac{2^{-\sigma}}{n-1}
\end{equation}
from which the lemma follows by the triangle inequality for statistical distance.

Let $\vec{x}$ and $\vec{x}'$ be related by a basic step and suppose w.l.o.g. that $\vec{x}$ and $\vec{x}'$ differ in the first two co-ordinates. Taking $m=2+5\lceil\log(q)\rceil+\lceil 2\sigma+2 \log(n-1) \rceil$, by lemma \ref{lemma:twocase}, we can couple the values sent by the first two parties on input $\vec{x}$ with the values they send on input $\vec{x}'$ so that they match with probability $1-\frac{2^{-\sigma}}{n-1}$. Independently of that we can couple the inputs of the other $n-2$ parties so that they always match as they each have the same input in both cases. This gives a coupling exhibiting that equation \ref{eq:basicbound} holds.
\end{proof}

\begin{remark}
  It may seem counter intuitive to require more messages the more parties there are (for fixed $q$). The addition of the $\log(n-1)$ term to $m$ is necessary for the proof of Lemma \ref{lemma:IKOS}. The is because we are trying to stop the adversary from learning a greater variety of things when we have more parties. However it may be the case that Theorem 4.1 could follow from a weaker guarantee than provided by Lemma \ref{lemma:IKOS} and such a property might be true without the presence of this term.

It is an open problem to prove a lower bound greater than two on the number of messages required to get $O(1)$ error on real summation. A proof that one message is not enough is given in \cite{DBLP:conf/crypto/BalleBGN19}.
\end{remark}

\subsection{Improving the constants}
\label{sec:improving}

  The constants implied by this proof can be improved by using a sharper bound for $\binom{2m}{m}$ in inequality \ref{al:binombound}. Using the bound $\binom{2m}{m}\geq \frac{4^m}{\sqrt{\pi(m+1/2)}}$ gives that taking $m$ to be the ceiling of the root of
\begin{equation*}
  m=1+\sigma+\frac{5\lceil \log(q)\rceil}{2}+\frac{1}{4}\log(\pi(m+\frac{1}{2}))
\end{equation*}
suffices in the statement of Lemma~\ref{lemma:twocase}. The resulting value of $m$ is
\begin{equation*}
  \frac{5}{2}\log(q)+\sigma+\frac{1}{4}\log(\log(q)+\sigma)+O(1) \enspace.
\end{equation*}
Adding $\log(n-1)$ to the root before taking the ceiling gives the following value of $m$ for which Lemma \ref{lemma:IKOS} holds
\begin{align*}
m = \frac{5}{2}\log(q)+\sigma+\log(n-1) + \frac{1}{4}\log(\log(q)+\sigma + \log(n-1))+O(1) \enspace.
\end{align*}
 
\section{Technical Lemmas}
\label{sec:technical}

\paragraph{Randomized rounding.} Our protocols use a fixed point encoding of a real number $x$ with integer precision $p>0$ and randomized rounding, which we define
as $\fp(x, p) = \lfloor xp \rfloor + \bernoulli(xp - \lfloor xp \rfloor)$.
We note this rounding is unbiased in the sense that $\Ex[\fp(x, p)]/p = x$.
The following lemma provides a simple bound on the MSE of this operation.

\begin{lemma}
  \label{lem:roundingerror}
For any  $\vec{x}\in\mathbb{R}^n$, $\mse(\sum_{i=1}^n \fp(x_i, p)/p, \sum_{i=1}^n x_i) \leq n/(4p^2)$.
\end{lemma}
\begin{proof}
  Let  $\Delta_i$ be $\fp(x_i, p)/p - x_i$, and note that $|\Delta_i|\leq 1/p$ and $\EE(\Delta_i)=0$. It follows that
  \begin{align*}
    \mse&\left(\sum_{i=1}^n \fp(x_i, p)/p, \sum_{i=1}^n x_i\right) = \EE\left[\left(\sum_{i=1}^n \Delta_i\right)^2\right] = \\
    &\sum_{i=1}^n \EE[\Delta_i^2] + \sum\limits_{1\leq i < j\leq n}(\EE[\Delta_i\Delta_j]) = \sum_{i=1}^n \EE[\Delta_i^2] \leq n/(4p^2).
  \end{align*}
\end{proof}

\paragraph{Differential privacy from total variation distance.}

The following lemma (also stated by Wang et al.~\cite{DBLP:conf/icml/WangFS15}, Proposition $3$) provides a convenient method to obtain differential privacy guarantees by comparing the output distributions two protocols in terms of total variation distance.
Recall that the total variation distance between two random variables $\rv{X}$ and $\rv{Y}$ can be defined as $\TV(\rv{X},\rv{Y}) = \sup_{E} | \Pr[\rv{X} \in E] - \Pr[\rv{Y} \in E]|$.

\begin{lemma}
\label{lemma:SD2delta}
Let $\mech{M}:\dom{X}^n \to \dom{O}$ and $\mech{M'}:\dom{X}^n \to \dom{O}$ be protocols such that
$\SD(\mech{M}(\tup{x}), \mech{M}'(\tup{x})) \leq \mu$ for all inputs $\tup{x} \in \dom{X}^n$.
If $\mech{M}$ is $(\epsilon, \delta)$-DP, then
$\mech{M}'$ is $(\epsilon, \delta + (1+e^\epsilon)\mu)$-DP. 
\end{lemma}
\begin{proof}
For any neighboring inputs $\tup{x}, \tup{x}' \in\dom{X}^n$ and $E \subseteq \dom{O}$ we have$\PP[\mech{M}(\tup{x}) \in E] \leq e^\epsilon\PP[\mech{M}(\tup{x}')\in E] + \delta$ and $|\PP[\mech{M}'(\tup{x}) \in E] - \PP[\mech{M}(\tup{x})\in E]| \leq \mu$. It follows that $\PP[\mech{M}'(\tup{x})\in E] \leq \PP[\mech{M}(\tup{x})\in E] + \mu \leq e^\epsilon\PP[\mech{M}(\tup{x}')\in E] + \delta + \mu \leq e^\epsilon(\PP[\mech{M}'(\tup{x}')\in E] + \mu) + \delta + \mu$.
\end{proof}

\section{Additional Experiments}
\label{sec:more-experiments}
In this section we present numerical evaluations of the communication and accuracy bounds of the protocols proposed in the paper and those that existed previously in the literature.

\subsection{Numerical Table}
\label{sec:numerics_table}

\begin{table*}
\footnotesize
\begin{center}
  \begin{tabular}{|c| c |c |c |c |c |c |c |c|} 
    \hline
     & \multicolumn{4}{|c|}{Number of messages} & \multicolumn{4}{|c|}{MSE of resultant sum} \\ [0.5ex]
    \hline
    Protocol & \multicolumn{2}{|c|}{$n=10^4$} & \multicolumn{2}{|c|}{$n=10^5$} & \multicolumn{2}{|c|}{$n=10^4$} & \multicolumn{2}{|c|}{$n=10^5$} \\ [0.5ex]
    \hline
     & $\epsilon=0.5$ & $\epsilon=1$ & $\epsilon=0.5$ & $\epsilon=1$ & $\epsilon=0.5$ & $\epsilon=1$ & $\epsilon=0.5$ & $\epsilon=1$ \\ [0.5ex] 
    \hline
    LocalDP & 1 & 1 & 1 & 1 & 41677.0 & 11706.7 & 416769.8 & 117067.4 \\ 
    \hline
    CuratorDP & -- & -- & -- & -- & 8.0 & 2.0 & 8.0 & 2.0 \\ 
    \hline
    CheuEtAl~\cite{DBLP:journals/corr/abs-1808-01394} & 50 & 100 & 159 & 317 & 693.5 & 206.3 & 888.0 & 252.4 \\ 
    \hline
    BalleEtAl~\cite{DBLP:conf/crypto/BalleBGN19} & 1 & 1 & 1 & 1 & 592.9 & 278.8 & 1433.4 & 683.8 \\ 
    \hline
    Recursive (2 msg) & 2 & 2 & 2 & 2 & 361.6 & 153.4 & 989.6 & 418.7 \\ 
    \hline
    Recursive (3 msg) & 3 & 3 & 3 & 3 & $\infty$ & 161.9 & 530.6 & 316.8 \\ 
    \hline
    Recursive (2 msg optimized) & 2 & 2 & 2 & 2 & 353.5 & 121.8 & 114.4 & 234.9 \\ 
    \hline
    Recursive (3 mgs optimized) & 3 & 3 & 3 & 3 & 354.8 & 128.8 & 197.0 & 152.4 \\ 
    \hline
    IKOS (Original) & 69 & 70 & 84 & 84 & 8.2 & 2.2 & 8.2 & 2.2 \\ 
    \hline
    IKOS (Improved) & 9 & 9 & 9 & 9 & 8.2 & 2.2 & 8.2 & 2.2 \\ 
    \hline
    IKOS (GhaziEtAl~\cite{DBLP:journals/corr/abs-1909-11073}) & 411 & 415 & 399 & 402 & 8.2 & 2.2 & 8.2 & 2.2 \\ 
    \hline
  \end{tabular}
\end{center}
\caption{Evaluation of the bounds on the number of messages per client and MSE of each of the protocols we have discussed. More details of each protocol and an interpretation of the results are provided in the text.}
\label{tab:numerics}
\end{table*}

For each protocols we find both the number of messages required and the resulting bound on the mean squared error of the sum. For each protocol we consider the case of $n\in \{10^4,10^5\}$ and $\epsilon\in \{0.5,1\}$. Where a protocol only provides approximate differential privacy we set $\delta=1/n^2$. These results are presented in Table~\ref{tab:numerics}.

The first two rows are randomized response\footnote{Implemented by using a standard binary randomized response after applying an unbiased randomized rounding to $x \in [0,1]$.} in the local model and the Laplace mechanism in the curator model. The curator model requires a trusted curator, so should have very good error, and due to it working on centralized data there is no number of messages. The local model on the other hand requires no trust from the parties so any method that failed to beat the error in that case would be useless.

The remaining methods are all methods in the shuffle model. They thus have intermediate security assumptions and the aim is to get error close to that of the curator model.

The next two methods, BalleEtAl and CheuEtAl, represent the state of the art before this work. Both in the single message model and the multi-message model, this was the baseline we had to work from before starting this work.

The next four rows are the recursive protocol form Section~\ref{sec:rec_protocol}. We show how well it performs with both two messages and three messages. We also present the protocol both with its parameters chosen by the expressions suggested by the asymptotic theorems and with parameters chosen by computer search to optimize the resulting error (see Section~\ref{sec:numerics_details} for details).

As the table shows, for the parameters we look at our two message protocol is already beating the error of BalleEtAl, even without the optimization. With the optimization (and usually without) it also beats the error of CheuEtAl using far fewer messages. As can be seen the three message protocol requires $n$ to be quite large and $\epsilon$ to be quite small in order to beat the two message protocol. This shouldn't be surprising as the theory recommends $O(\log\log(n))$ messages, this indicates that the recursive protocol will not benefit from a fourth message in any practical setting.

The final three protocols are those based on the protocol of Ishai et al.~\cite{ikos}. They all work by adding noise using the P\'olya distribution to that protocol. The difference is in the analysis and thus required number of messages. The ``original'' version uses the analysis of Ishai et al. paper~\cite{ikos} (cf.~Section~\ref{sec:secure_sum}). The improved version is the new analysis of that algorithm in this paper (Corollary~\ref{cor:wst-security}). The table shows the improvement is substantial for practical values of $n$, not just asymptotically. The final row displays the number of messages required according to the analysis of Ghazi et al.~\cite{DBLP:journals/corr/abs-1909-11073} (see also Section~\ref{sec:concurrent}), this shows that they were interested in asymptotics and not in constants. Note that the error here is very close to the curator model.

\subsection{Implementation Details}\label{sec:numerics_details}

These values were found with python code which can be found at
\url{https://github.com/adriagascon/shuffledpsummation}.
Mostly, this code merely evaluated the bounds given in theorems in this and other papers. There are two exceptions to this.

Firstly, the algorithm of Balle et al.~\cite{DBLP:conf/crypto/BalleBGN19} has a parameter that affects security and must be optimized. We sharply optimize this using code that they provide, both for the BalleEtAl entry and the recursive entries.

Secondly, in the recursive algorithm there are parameters $p_i$ and $\epsilon_i$ to choose. We present the results of using the expressions given in the theorems and of using optimized versions. To optimize these parameters one could cycle over all feasible values of $p_i$ and (to some precision) $\epsilon_i$, however this is computationally prohibitive. Instead we assumed that the error was convex in each of the $p_i$ and $\epsilon_i$, this enabling convex optimization techniques, i.e. golden section search, to find the best values. We stress that if these assumptions turned out to be false, that might mean that we hadn't found the optimal parameters, but the errors given in the table would still be valid bounds on the error for the parameter choices we found.

\end{document}